\pdfoutput=1
\documentclass[11pt]{article}
\usepackage[letterpaper, left=1in, right=1in, top=1in,bottom=1in]{geometry}

\PassOptionsToPackage{linesnumbered,noend}{algorithm2e} 
\PassOptionsToPackage{bookmarks=false,pagebackref=true,colorlinks=true,citecolor=violet,linkcolor=blue}{hyperref} 
\PassOptionsToPackage{numbers}{natbib}

\newcount\Comments  
\newcommand{\ignore}[1]{}
\newcommand{\kibitz}[2]{\ifnum\Comments=1\textcolor{#1}{#2}\fi}

\usepackage{parskip}
\makeatletter
\def\thm@space@setup{%
  \thm@preskip=\parskip \thm@postskip=2pt
}
\makeatother

\usepackage[T1]{fontenc}
\usepackage{amsmath,amsfonts,amsthm,amssymb}
\usepackage{enumitem}
\usepackage{bbm}
\usepackage[numbers]{natbib}
\usepackage{booktabs}
\usepackage{hyperref}
\usepackage[bookmarks=false,pagebackref=true,colorlinks=true,citecolor=violet,linkcolor=blue]{hyperref}
\usepackage{nicefrac}
\usepackage{tikz}
\usetikzlibrary{shapes.misc}

\tikzset{cross/.style={cross out, draw=black, minimum size=2*(#1-\pgflinewidth), inner sep=0pt, outer sep=0pt},
cross/.default={2pt}}
\usepackage{wrapfig}
\usepackage[ruled]{algorithm2e}

\SetCommentSty{mycommfont}
\usepackage{mathtools}
\usepackage{multirow}
\usepackage{subcaption}
\usepackage[makeroom]{cancel}
\usepackage{bigints}
\usepackage{relsize}
\usepackage[title]{appendix}
\usepackage{pgf,tikz}
\usetikzlibrary{patterns}
\usetikzlibrary{arrows}
\usepackage{mathrsfs}
\usepackage{mdframed}
\DeclareMathOperator{\sgn}{sgn}

\DeclareMathOperator*{\E}{\mathbb{E}}
\let\Pr\relax
\DeclareMathOperator*{\Pr}{\mathbb{P}}
\DeclarePairedDelimiter\abs{\lvert}{\rvert}
\DeclarePairedDelimiter\norm{\lVert}{\rVert}

\usepackage{xspace}

\newcommand{\eps}{\varepsilon}
\newcommand{\bbR}{\mathbb{R}}
\newcommand{\bbN}{\mathbb{N}}
\newcommand{\vu}{\mathbf{u}}
\newcommand{\vx}{\mathbf{x}}
\newcommand{\vy}{\mathbf{y}}
\newcommand{\vz}{\mathbf{z}}
\newcommand{\tp}{\tilde{p}_\delta}
\newcommand{\Prin}{\mathbb{P}_t^{\texttt{\upshape in}}}
\newcommand{\dBMR}{\texttt{\upshape{$\delta$-BMR}}\xspace}
\newcommand{\alphast}{\alpha^{\star}}
\newcommand{\vr}{\mathbf{r}}
\newcommand{\dist}{\texttt{sdist}}

\newcommand{\hy}{\hat{y}}

\newcommand{\up}{\underline{p}_\delta}

\newcommand{\calX}{\mathcal{X}}

\newcommand{\calP}{\mathcal{P}}
\newcommand{\bcalP}{\overline{\mathcal{P}}}
\newcommand{\calA}{\mathcal{A}}
\newcommand{\calB}{\mathcal{B}}
\newcommand{\calD}{\mathcal{D}}
\newcommand{\calN}{\mathcal{N}}

\newcommand{\calR}{\mathcal{R}}

\newcommand{\calO}{\mathcal{O}}

\newcommand{\tPr}{\widetilde{\Pr}}

\newcommand{\hell}{\hat{\ell}}
\newcommand{\talpha}{\tilde{\alpha}}

\newcommand{\1}{\mathbbm{1}}
\newcommand{\Bern}{\mathtt{Bern}}
\newcommand{\f}{f_{t}}
\newcommand\numberthis{\addtocounter{equation}{1}\tag{\theequation}}
\newcommand{\grinder}{\textsc{Grinder}\xspace}
\newcommand{\vw}{\mathbf{w}}
\newcommand{\kl}{\texttt{KL}}

\newcommand{\squishlist}{
   \begin{list}{$\bullet$}
    { \setlength{\itemsep}{0pt}      \setlength{\parsep}{3pt}
      \setlength{\topsep}{3pt}       \setlength{\partopsep}{0pt}
      \setlength{\leftmargin}{1.5em} \setlength{\labelwidth}{1em}
      \setlength{\labelsep}{0.5em} } }
\newcommand{\squishend}{  \end{list}  }

\newtheorem{theorem}{Theorem}[section]
\newtheorem{lemma}[theorem]{Lemma}
\newtheorem{definition}[theorem]{Definition}

\newtheorem{proposition}[theorem]{Proposition}

\usepackage[suppress]{color-edits}
\usepackage{color-edits}
\addauthor{yc}{violet}
\addauthor{yl}{blue}
\addauthor{cp}{red}

\begin{document}
\title{Learning Strategy-Aware Linear Classifiers}

\author{Yiling Chen \\
        Harvard University \\
        \texttt{yiling@seas.harvard.edu} \\  
        \and Yang Liu \\ 
        UC Santa Cruz \\
        \texttt{yangliu@ucsc.edu} \\
        \and Chara Podimata \\ 
        Harvard University \\ 
        \texttt{podimata@g.harvard.edu}\\
        }

\maketitle

\begin{abstract}%
We address the question of repeatedly learning linear classifiers against agents who are \emph{strategically} trying to \emph{game} the deployed classifiers, and we use the \emph{Stackelberg regret} to measure the performance of our algorithms. First, we show that Stackelberg and external regret for the problem of strategic classification are \emph{strongly incompatible}: i.e., there exist worst-case scenarios, where \emph{any} sequence of actions providing \emph{sublinear} external regret might result in \emph{linear} Stackelberg regret and vice versa. Second, we present a strategy-aware algorithm for minimizing the Stackelberg regret for which we prove nearly matching upper and lower regret bounds. Finally, we provide simulations to complement our theoretical analysis. Our results advance the growing literature of learning from revealed preferences, which has so far focused on ``smoother'' assumptions from the perspective of the learner and the agents respectively. 

\end{abstract}

\section{Introduction}\label{sec:intro}

As Machine Learning (ML) algorithms become increasingly involved in real-life decision making, the agents that they interact with tend to be neither stochastic nor adversarial. Rather, they are \emph{strategic}. For example, consider a college that wishes to deploy an ML algorithm to make admissions decisions. Student candidates might try to manipulate their test scores in an effort to fool the classifier. Or think about email spammers who are trying to manipulate their emails in an effort to fool the ML classifier and land in the non-spam inboxes. Importantly, in both examples the agents (students and spammers respectively) do not want to sabotage the classification algorithm only for the sake of harming its performance. They merely want to game it for their own benefit. And this is precisely what differentiates them from being fully adversarial.

Motivated by the problem of classifying spam emails, we focus on the problem of learning an unknown \emph{linear} classifier, when the training data come in an \emph{online} fashion from \emph{strategic} agents, who can alter their feature vectors to \emph{game} the classifier. We model the interplay between the learner and the strategic agents\footnote{We refer to the learner as a female (she/her/hers) and to the agents as male (he/his/him).} as a repeated \emph{Stackelberg game} over $T$ timesteps. In a repeated Stackelberg game, the learner (``leader'') \emph{commits} to an action, and then, the agent (``follower'') best-responds to it, i.e., reports something that maximizes his underlying utility. The learner's goal is to minimize her \emph{Stackelberg regret}, which is the difference between her cumulative loss and the cumulative loss of her best-fixed action in hindsight, \emph{had she given the agent the opportunity to best-respond to it}.

{\bf Our Contributions.}
\squishlist
\item We study a general model of learning interaction in strategic classification settings where the agents' true datapoint remains hidden from the learner, the agents can misreport within a ball of radius $\delta$ of their true datapoint (termed \emph{$\delta$-bounded, myopically rational (\dBMR) agents}), and the learner measures her performance using the \emph{binary loss}. This model departs significantly from the smooth utility and loss functions used so far for strategic classification (Sec.~\ref{sec:model}).
\item We prove that in strategic classification settings against \dBMR agents \emph{simultaneously} achieving sublinear external and Stackelberg regret is in general impossible (\emph{strong incompatibility}) (i.e., application of standard no-external regret algorithms might be unhelpful (Sec.~\ref{sec:regret-notions})). 
\item Taking advantage of the structure of the responses of \dBMR agents while working in the \emph{dual} space of the learner, we propose an adaptive discretization algorithm (\grinder), which uses access to an oracle. \grinder's novelty is that it assumes no stochasticity for the adaptive discretization (Sec.~\ref{sec:grinding-algo}). 
\item We prove that the regret guarantees of \grinder remain \emph{unchanged} in order even when the learner is given access to a \emph{noisy} oracle, accommodating more settings in practice (Sec.~\ref{sec:grinding-algo}).
\item We prove nearly matching lower bounds for strategic classification against \dBMR agents (Sec.~\ref{sec:lb}).
\item We provide simulations implementing \grinder both for continuous and discrete action spaces, and using both an accurate and an approximation oracle (Sec.~\ref{sec:sims}). %
\squishend

{\bf Our Techniques.}
\squishlist
\item In order to prove the incompatibility results of the regret notions in strategic classification, we present a formal framework, which may be of independent interest.
\item To overcome the non-smooth utility and loss functions, we work on the \emph{dual} space, which provides information about various \emph{regions} of the learner's action space, \emph{despite never observing the agent's true datapoint}. These regions (polytopes) relate to the partitions that \grinder creates.
\item To deal with the learner's action space being \emph{continuous} (i.e., containing infinite actions), we use the fact that \emph{all actions} within a polytope share the same \emph{history} of estimated losses. So, passing information down to a recently partitioned polytope becomes a simple volume reweighting. 
\item To account for all the actions in the continuous action space, we present a formulation of the standard \texttt{EXP3} algorithm that takes advantage of the polytope partitioning process. 
\item For bounding the variance of our polytope-based loss estimator, we develop a polytope-based variant of a well-known graph-theoretic lemma (\citep[Lem.~5]{NCBDK15}), which has been crucial in the analysis of online learning settings with feedback graphs. Such a variant is mandatory, since direct application of \citep[Lem.~5]{NCBDK15} in settings with continuous action spaces yields vacuous\footnote{Due to the logarithmic dependence on the number of actions.} regret. 
\item We develop a generalization of standard techniques for proving regret lower bounds in \emph{strategic} settings, where the datapoints that the agents report \emph{change} in response to the learner's actions.
\squishend

{\bf Related Work.} Our work is primarily related to the literature on \emph{learning using data from strategic sources} (e.g., \citep{CIL15,CDP15,PP04,DFP10,CPPS18,BPT17,BPT18,WC09,BS11,MAMR11,MPR12}). %
Our work is also related to learning in Stackelberg Security Games (SSGs) (\citep{LCM09,MTS12,BHP14}) and especially, the work of \citet{BBHP15}, who study information theoretic sublinear Stackelberg regret\footnote{Even though the formal definition of Stackelberg regret was only later introduced by \citet{drsww18}.} algorithms for the learner. In SSGs, all utilities are linear, a property not present in strategic classification against \dBMR agents. Finally, our work is related to the literature in online learning with \emph{partial} (see \citep{BCB12,S19,LS19}) and \emph{graph-structured} feedback \citep{NCBDK15,CHK16}. Adaptive discretization algorithms were studied for stochastic Lipschitz settings in \citep{KSU08,BMSS11}, but in learning against \dBMR agents, the loss is neither stochastic nor Lipschitz. %

\section{Model and Preliminaries}\label{sec:model}

We use the spam email application as a running example to setup our model. Each agent has an email that is either a spam or a non-spam. Given a classifier, the agent can alter his email to a certain degree in order to bypass the email filter and have his email be classified as non-spam. Such manipulation is costly. Each agent chooses a manipulation to maximize his overall utility. %

\paragraph{Interaction Protocol.} Let $d \in \bbN$ denote the dimension of the problem %
and $\calA \subseteq [-1,+1]^{d+1}$ the learner's action space\footnote{This is wlog, as the normal vector of any hyperplane can be normalized to lie in $[-1,1]^{d+1}$.}. Actions $\alpha \in \calA$ correspond to hyperplanes%
, written in terms of their normal vectors, and we assume that the $(d+1)$-th coordinate encodes information about the intercept. Let $\calX \subseteq ([0,1]^d,1)$ the feature vector space%
, where by $([0,1]^d,1)$ we denote the set of all $(d+1)$-dimensional vectors with values in $[0,1]$ in the first $d$ dimensions and value $1$ in the $(d+1)$-th. Each feature vector has an associated label $y \in \mathcal{Y} = \{-1,+1\}$.  
Formally, the interaction protocol (which repeats for $t \in [T]$) is given in Protocol~\ref{protocol}, where by $\sigma_t$ we denote the tuple (feature vector, label). 

\begin{algorithm}[htbp]
\DontPrintSemicolon
\SetAlgorithmName{Protocol}{Protocol}
\SetAlgoLined
\caption{Learner-Agent Interaction at Round $t$}\label{protocol} 
Nature adversarially selects feature vector $\vx_t \in \calX \subseteq ([0,1]^d, 1)$. \hspace{-0.05in}\tcp*{agent's original email}
The learner chooses action $\alpha_t \in \calA$, and commits to it. \tcp*{learner's linear classifier}
Agent observes $\alpha_t$ and $\sigma_t = (\vx_t, y_t)$, where $y_t \in \mathcal{Y}$. ~\hspace{-0.1in}\tcp*{$y_t=+1$, if non-spam originally} 
Agent reports feature vector $\vr_t( \alpha_t, \sigma_t ) \in \calX$ (potentially, $\vr_t( \alpha_t, \sigma_t ) \neq \vx_t$). \label{step:protocol-4} \tcp*{altered email}
The learner observes $(\vr_t(\alpha_t, \sigma_t), \hat{y}_t)$, where $\hat{y}_t \in \mathcal{Y}$ is the label of $\vr_t(\alpha_t, \sigma_t)$, and incurs binary loss $\ell(\alpha_t, \vr_t( \alpha_t, \sigma_t ), \hy_t) = \1 \{ \sgn(\hat{y}_t \cdot \langle \alpha_t, \vr_t( \alpha_t, \sigma_t ) \rangle ) = -1 \}$. \tcp*{loss on {\bf altered} email}
\end{algorithm}%

{\bf Agents' Behavior: $\delta$-Bounded Myopically Rational.} Drawing intuition from the email spam example, we focus on agents who can alter their feature vector \emph{up to an extent} in order to make their email fool the classifier, and, if successful, they gain some value. The agent's reported feature vector $\vr_t(\alpha_t, \sigma_t)$ is the solution to the following constrained optimization problem\footnote{For simplicity, we denote $\vr_t(\alpha_t) = \vr_t(\alpha_t, \sigma_t)$ when clear from context.}: 
\begin{align*}
\vr_t(\alpha_t, \sigma_t) = \arg \max_{\|\vz - \vx_t\| \leq \delta} u_t(\alpha_t, \vz, \sigma_t)
\end{align*}
where $u_t(\cdot, \cdot, \cdot)$ is the agent's underlying utility function, \emph{which is unknown to the learner}. In words, in choosing what to report, the agents are \emph{myopic} (i.e., focus only on the current round $t$), \emph{rational} (i.e., maximize their utility), and \emph{bounded} (i.e., misreport in a ball of radius $\delta$ around $\vx_t$). In such settings (e.g., email spam example), the agents derive no value if, by altering their feature vector from $\vx_t$ to $\vr_t(\alpha_t)$, they also change $y_t$. Indeed, a spammer ($y_t = -1$) wishes to fool the learner's classifier, without actually having to change their email to be a non-spam one ($\hat{y}_t = +1$). Since the agents are \emph{rational} this means that the observed label by the learner is $\hy_t = y_t$ and we only use notation $y_t$ for the rest of the paper. We call such agents $\delta$-Bounded Myopically Rational (\dBMR)\footnote{Note that if the agents were \emph{adversarial}, they would report $\vr(\alpha_t) = \arg\max_{\|\vz - \vx_t \| \leq \delta} \ell(\alpha, \vz, y_t)$.}.

Note that \dBMR agents include (but are not limited to!) a broad class of rational agents for strategic classification, like for example agents whose utility is defined as:
\begin{equation}\label{eq:examp-dBMR}
u_t(\alpha_t, \vr_t(\alpha_t, \sigma_t), \sigma_t) = \delta' \cdot \1 \left\{\sgn(\langle \alpha_t, \vr_t(\alpha_t, \sigma_t) \rangle) = +1 \right\} - \|\vx_t - \vr_t(\alpha_t, \sigma_t)\|_2
\end{equation}
where $\sgn(x) = +1$ if $x \geq 0$ and $\sgn(x) = -1$ otherwise. According to the utility presented in Eq.~\eqref{eq:examp-dBMR}, the agents get a value of $\delta'$ if he gets labeled as $+1$ and incurs a cost (i.e., time/resources spent for altering the original $\vx_t$) that is a metric. For this case, we have that $\delta \leq \delta'$.

\paragraph{Model Comparison with Other Strategic Classification Works.} Learning in strategic classification settings was studied in an offline model by \citet{HMPW16}, and subsequently, by \citet{drsww18} in an online model. Similar to our model, in \citep{HMPW16, drsww18} the ground truth label $y_t$ remains \emph{unchanged} even \emph{after} the agent's manipulation. Moreover, the work of \citet{drsww18} is orthogonal to ours in one key aspect: they find the appropriate conditions which can guarantee that the best-response of an agent, written as a function of the learner's action, is concave. As a result, in their model the learner's loss function becomes \emph{convex} and well known online convex optimization algorithms could be applied (e.g., \citep{FKM05,BLE17}) in conjunction with the mixture feedback that the learner receives. The foundation of our work, however, is settings with less smooth utility and loss functions for the agents and the learner respectively, where incompatibility issues arise. There has also been recent interest in strategic classification settings where the agents by misreporting actually end up changing their label $y_t$ \citep{BLWZ20, SEA20, PZMH20}. These models are especially applicable in cases where in order to alter their feature vector $\vx_t$ (e.g., qualifications for getting in college) the agents have to \emph{improve} their ground truth label (e.g., actually \emph{try} to become a better candidate \citep{ustun2019actionable}). In contrast, in our work we think of the misreports as ``manipulations'' that aim at gaming the system without altering $y_t$.

\section{Stackelberg versus External Regret}\label{sec:regret-notions}

For what follows, let $\{ \alpha_t\}_{t=1}^T$ be the sequence of the learner's actions in a repeated Stackelberg game. The full proof of this section can be found in Appendix~\ref{app:regr2}, and Appendices~\ref{app:pure-sg} and~\ref{app:loss-func} include detailed discussions around external and Stackelberg regret for learning in Stackelberg games.

\begin{definition}[External]
$R(T) = \sum_{t \in [T]} \ell(\alpha_t, \vr_t(\alpha_t), y_t) - \min_{\alphast_E \in \calA}\sum_{t \in [T]} \ell(\alphast_E, \vr_t(\alpha_t), y_t)$.
\end{definition} 
The external regret compares the cumulative loss from $\{\alpha_t\}_{t \in [T]}$ to the cumulative loss incurred by the best-fixed action in hindsight, had learner \emph{not} given the opportunity to the agents to best respond. 
\begin{definition}[Stackelberg]\begin{small}
$\calR(T) = \sum_{t \in [T]} \ell(\alpha_t, \vr_t(\alpha_t),y_t) - \min_{\alphast \in \calA}\sum_{t \in [T]} \ell(\alphast, \vr_t(\alphast), y_t)$.\end{small}
\end{definition}
\begin{wrapfigure}{R}{0.35\textwidth}
    \centering
    \begin{tikzpicture}[line cap=round,line join=round,>=latex',x=3cm,y=3cm, scale=1]
        \draw[->,color=black] (-0.1,0.) -- (1.2,0.); %
        \foreach \x in {-0.1, 0.1, 0.2, 0.3, 0.4, 0.5, 0.6, 0.7, 0.8, 0.9, 1}
        \draw[shift={(\x,0)},color=black] (0pt,2pt) -- (0pt,-2pt) node[below] {};
        \draw[->,color=black] (0.,-0.1) -- (0.,1.2); %
        \foreach \y in {-0.1, 0.1, 0.2, 0.3, 0.4, 0.5, 0.6, 0.7, 0.8, 0.9, 1}
        \draw[shift={(0,\y)},color=black] (2pt,0pt) -- (-2pt,0pt) node[left] {};
        \draw[color=black, fill=black] (0.8,1) circle (1.5pt);
        \draw[color=black, fill=black] (0.65,0.3) circle (1.5pt);
        \draw[color=black, fill=black] (0.6,0.6) circle (1.5pt);
        \draw[color=black, fill=black] (0.4,0.5) circle (1.5pt);
        \draw[dotted, color=black] (0.8,1)    circle (0.3cm); %
        \draw[dotted, color=black] (0.65,0.3) circle (0.3cm); %
        \draw[dotted, color=black] (0.6,0.6)  circle (0.3cm); %
        \draw[dotted, color=black] (0.4,0.5)  circle (0.3cm); %
        \draw [color=red, semithick]   (-0.1, 1.1) -- (1.1, -0.1); 
        \draw [color=blue, semithick]           (-0.1, 0.2) -- (1.1, 0.8); 
        \draw [dotted, color=red]   (0.65,0.3)  -- (0.72, 0.38);
        \draw [dotted, color=blue]   (0.6,0.6)  -- (0.69, 0.5);
        \draw[dotted, color=red] (0.4, 0.5)   -- (0.43, 0.6);
        \draw[color=red, fill=red] (0.69,0.37) -- (0.73, 0.37) -- (0.71, 0.40);
        \draw[color=blue, fill=blue] (0.65,0.50) rectangle (0.69,0.54);
        \draw[color=red, fill=red] (0.41, 0.59) -- (0.45, 0.59) -- (0.43, 0.62) ;
        \draw[color=blue, fill=blue] (0.37, 0.38) rectangle (0.41, 0.42);
        \draw[->, color=red] (0.9, 0.1) -- (1, 0.2);
        \draw[->, color=blue] (1, 0.75) -- (1.07, 0.6);
        \begin{footnotesize}
        \draw[color=red]  (1, -0.1) node {$h$};
        \draw[color=black]  (1.2, -0.05) node {$x_1$};
        \draw[color=black]  (-0.06, 1.2) node {$x_2$};
        \draw[color=blue]  (1, 0.85) node {$h'$};
        \draw[color=red]   (1, 0.1) node {$+1$};
        \draw[color=blue]  (1.15, 0.7) node {$+1$};
        \draw[color=black]  (0.63, 0.23) node {$\vx^4$};
        \draw[color=black]  (0.75, 1.05) node {$\vx^3$};
        \draw[color=red]     (0.3,  1.02) node {$\vr_t=$};
        \draw[color=blue]    (0.55,  1.02) node {$\vr_t=$};
        \draw[color=blue]    (0.45,  0.21) node {$\vr_t=$};
        \draw[color=red]     (0.75,  0.59) node {$=\vr_t$};
        \draw[color=red]     (0.8,  0.4)  node  {$\vr_t$};
        \draw[color=red]     (0.4,  0.7)  node  {$\vr_t$};
        \draw[color=blue]    (0.33,  0.35)  node  {$\vr_t$};
        \draw[color=blue]     (0.8,  0.5)  node  {$\vr_t$};
        \draw[color=black]  (0.6, 0.67) node {$\vx^2$};
        \draw[color=black]  (0.35, 0.55) node {$\vx^1$};
        \end{footnotesize}
        \end{tikzpicture}
    \caption{Black dots denote true feature vectors. Axes $x_1, x_2$ correspond to the two features. Dotted circles correspond to the $\delta$-bounded interval inside which agents can misreport. Blue squares correspond to misreports against action $h'$ and red triangles to misreports against action $h$. \vspace{-20pt}}\label{fig:incomp-ext-stack}
\end{wrapfigure}
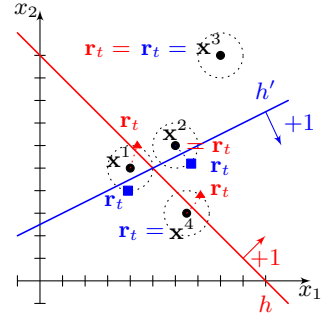%
Stackelberg regret \citep{BBHP15,drsww18} %
compares the loss from $\{\alpha_t\}_{t \in [T]}$ to the loss from the best-fixed action in hindsight, \emph{had learner given the opportunity to the agents to best respond}. 

\begin{theorem}\label{thm:incomp-sr-ext}
There exists a repeated strategic classification setting against a \dBMR agent, where \emph{every} action sequence from the learner with \emph{sublinear} external regret incurs \emph{linear} Stackelberg regret, and \emph{every} action sequence for the learner with \emph{sublinear} Stackelberg regret incurs \emph{linear} external regret.
\end{theorem}

\vspace{-7pt}
\begin{proof}[Proof Sketch]
We construct the following instance of an online strategic classification setting against \dBMR agents (pictorially shown in Figure~\ref{fig:incomp-ext-stack}). Let the action space be $\calA = \{h, h'\}$ such that $h = (1,1,-1)$ and $h' = (0.5, -1, 0.25)$, and let $\delta = 0.1$. Nature draws feature vectors $\vx^1 = (0.4, 0.5, 1), \vx^2 = (0.6, 0.6, 1), \vx^3 = (0.8, 0.9, 1), \vx^4 = (0.65,0.3, 1)$ with probabilities $p^1=0.05, p^2=0.15, p^3=0.05, p^4=0.75$, and with labels $y^1=-1, y^2=-1, y^3=+1, y^4=+1$. Note that these original feature vectors are even \emph{separable} by a \emph{margin}! The expected loss for each action $\alpha \in \calA$ corresponds to the number of mistakes that the learner makes against $\vr_t(\alpha)$, which in turn depends on the probability with which nature drew each of the feature vectors $\vx^1, \vx^2, \vx^3, \vx^4$, e.g., $\E[ \ell (h, \vr_t( h),y_t) ] = 0.2$ because classifier $h$ makes a mistake for points $\vx^1$ and $\vx^2$. Analogously, $\E[ \ell (h', \vr_t( h' ),y_t)] = 0.25$, $\E[ \ell (h, \vr_t( h'),y_t) ] = 0.9$ and $\E[ \ell (h', \vr_t( h ),y_t) ] = 0.05$.

Every action sequence that yields a sublinear Stackelberg regret in this instance, must include action $h$ \emph{at least} $T - o(T)$ times (because $\E[ \ell (h, \vr_t( h ),y_t)] < \E[ \ell (h', \vr_t( h' ),y_t)]$), thus incurring cumulative \emph{loss} $0.2(T - o(T)) + 0.9 o(T)$. For such sequences, because responses $\vr_t(h)$ appear at least $T - o(T)$ times, the best-fixed action in hindsight for the external regret is action $h'$, with cumulative loss: $0.05(T - o(T))$. This means that the external regret is at least $0.15T$, i.e., \emph{linear}. For the next part of the proof, we show that if action $h'$ is played $T - o(T)$ times, then, the external regret is \emph{sublinear}. This is enough to prove our theorem, since in this case the Stackelberg regret is $0.05T$, i.e., \emph{linear}. If $h'$ is played $T - o(T)$ times, then the cumulative loss incurred is $0.05o(T) + 0.25(T-o(T))$ and the best fixed action in hindsight for the external regret is also action $h'$ with a cumulative loss $0.25(T - o(T)) + 0.05o(T)$. In other words, the external regret in this case is \emph{sublinear}. 
\end{proof}

\section{The $\grinder$ Algorithm}\label{sec:grinding-algo}

In this section, we present $\grinder$, an algorithm that learns to adaptively partition the learner's action space according to the agent's responses. To assist with the dense notation, we include notation tables in Appendix~\ref{app:notation-tables}. Formally, we prove the following\footnote{Our actual bound is \emph{tigher}, but harder to interpret without analyzing the algorithm.} \emph{data-dependent} performance guarantee.

\begin{theorem}\label{thm:regr-grind}
Given a finite horizon $T$ the Stackelberg regret incurred by $\grinder$ (Algo.~\ref{algo:grinding}) is: 
\begin{equation*}
\calR(T) \leq \calO \left( \sqrt{T \cdot \log \left(T\cdot\frac{\lambda\left( \calA \right)}{\lambda\left(\up\right)}\right) \cdot \log \left(\frac{\lambda\left( \calA \right)}{\lambda\left(\up\right)}\right)}\right)
\end{equation*} 
where by $\lambda(A)$ we denote the Lebesgue measure of any measurable space $A$, and by $\up$ we denote the polytope with the smallest Lebesgue measure that is induced by $\grinder$ after $T$ rounds' partition. 
\end{theorem}

\begin{wrapfigure}{R}{0.3\textwidth}
    \centering
    \begin{tikzpicture}[line cap=round,line join=round,>=latex',x=20.0cm,y=20.0cm, scale=0.3]
        \draw [color=white, xstep=2.0cm,ystep=2.0cm] (-0.1,-0.1) grid (0.5,0.4);
        \draw[->,color=black] (-0.1,0.) -- (0.5,0.); %
        \foreach \x in {-0.1, 0.1, 0.2, 0.3, 0.4}
        \draw[shift={(\x,0)},color=black] (0pt,2pt) -- (0pt,-2pt) node[below] {};
        \draw[->,color=black] (0.,-0.1) -- (0.,0.4); %
        \foreach \y in {-0.1, 0.1, 0.2, 0.3}
        \draw[shift={(0,\y)},color=black] (2pt,0pt) -- (-2pt,0pt) node[left] {};
        \draw[color=black] (0pt,-10pt) node[right] {};
        \draw [line width=1pt,color=red] (-0.05,0.35) -- (0.35,-0.1);
        \draw [line width=1pt, dashed]      (-0.05, 0.33) -- (0.44, 0.33); 
        \draw [line width=1pt, dashed]      (-0.07, -0.1) -- (0.05, 0.37); 
        \draw [line width=1pt, color=blue, dotted] (0.2,0.05) circle (2.cm);
        \draw [line width=1pt,color=green!50!black] (0.24,0.09) circle (4.cm);
        \draw [->,color=green!50!black] (0.24,0.085) -- (0.435,0.13);
        \draw [->,color=blue] (0.2,0.05) -- (0.14,-0.028);
        \begin{footnotesize}
        \draw[color=red] (-0.04,0.3) node {$\alpha$};
        \draw[color=black] (0.47, 0.33) node {$\beta$};
        \draw[color=black] (0.07, 0.37) node {$\gamma$};
        \draw[color=black]  (0.5, -0.03) node {$x_1$};
        \draw[color=black]  (-0.04, 0.4) node {$x_2$};
        \draw [color=blue,fill=blue] (0.2,0.05) circle (3pt);
        \draw[color=blue] (0.22,0.02) node {$\vx_t$};
        \draw[color=blue] (0.15,0.11) node {};
        \draw [color=green!50!black,fill=green!50!black] (0.24,0.085) circle (3pt);
        \draw[color=green!50!black] (0.225,0.12) node {$\vr_t(\alpha)$};
        \draw[color=green!50!black] (0.14,0.24) node {};
        \draw[color=green!50!black] (0.35,0.08) node {$2\delta$};
        \draw[color=blue] (0.16,0.03) node {$\delta$};
        \end{footnotesize}
        \end{tikzpicture}
    \caption{Agent's action space with axes $x_1, x_2$ corresponding to features. $\vx_t$ is the agent's true feature vector and $\vr_t(\alpha)$ his misreport against $\alpha$. Actions $\alpha, \beta, \gamma$ comprise the learner's action set. \vspace{-15pt}}\label{fig:2delta-ball}
\end{wrapfigure}
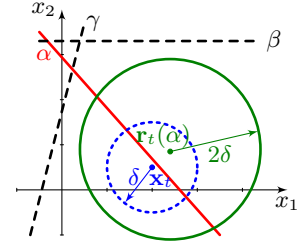

\vspace{-5pt}
\paragraph{Inferring $\ell(\alpha,\vr_t(\alpha),y_t)$ without Observing $\vx_t$.}We think of the learner's and the agent's spaces as dual ones (Fig.~\ref{fig:2delta-ball}), and focus on the agent's action space first. Since agents are \dBMR, then, for feature vector $\vx_t$ the agent can only misreport within the ball of radius $\delta$ centered around $\vx_t$, denoted by $\calB_\delta(\vx_t)$ (e.g., purple dotted circle in Fig.~\ref{fig:2delta-ball}). Since the learner observes $\vr_t(\alpha)$ and knows that the agent misreports in a ball of radius $\delta$ around $\vx_t$ (which remains unknown to her), she knows that in the worst case the agent's $\vx_t$ is found within the ball of radius $\delta$ centered at $\vr_t(\alpha)$. This means that the set of all of the agent's possible misreports against any committed action $\alpha'$ from the learner $\vr_t(\alpha' )$ is the \emph{augmented} $2\delta$ ball (e.g., green solid circle). Since $y_t$ is also observed by the learner, she can thus infer her loss $\ell(\alpha', \vr_t(\alpha'),y_t)$ for \emph{any} action $\alpha'$ that has $\calB_{2\delta}(\vr_t(\alpha))$ \emph{fully} in one of its halfspaces (e.g., actions $\beta, \gamma$ in Fig.~\ref{fig:2delta-ball}). 

In the learner's action space, actions $\alpha, \beta, \gamma$ are multidimensional \emph{points}, and this has a nice mathematical translation. An action $\gamma$ has $\calB_{2\delta}(\vr_t(\alpha))$ fully in one of its halfspaces, if its distance from $\vr_t(\alpha)$ is \emph{more than} $2\delta$. Alternatively for actions $h \in \calA$ such that: \[\frac{\abs*{\langle h, \vr_t(\alpha) \rangle}}{\|h\|_2} \leq 2\delta \Leftrightarrow \abs*{\langle h, \vr_t(\alpha) \rangle} \leq 4\sqrt{d} \delta\]where the last inequality comes from the fact that $\calA \subseteq [-1,1]^{d+1}$, the learner cannot infer $\ell(h, \vr_t(h))$. 
But for all other actions $\gamma$ in $\calA$, the learner can compute her loss $\ell(h, \vr_t(h))$ \emph{precisely}!
From that, we derive that the learner can partition her action space into the following \emph{polytopes}: upper polytopes $\calP_t^u$, containing actions $\vw \in \calA$ such that $\langle \vw, \vr_t(\alpha) \rangle \geq 4\sqrt{d} \delta$ and lower polytopes $\calP_t^l$, containing actions $\vw' \in \calA$ such that $\langle \vw', \vr_t(\alpha) \rangle \leq -4 \sqrt{d} \delta$. The distinction into the two sets is helpful as one of them always assigns label $+1$ to the agent's best-response, and the other always assigns label $-1$. The sizes of $\calP_t^u$ and $\calP_t^l$ depend on $\delta$ and $\{\vx_t\}_{t=1}^T$, but we omit these for the ease of notation.

\vspace{-7pt}
\begin{wrapfigure}{R}{0.3\textwidth}
    \centering
    \begin{tikzpicture}[line cap=round,line join=round,>=latex',x = 1cm, y = 1cm, scale=0.35]
        \draw [color=white,, xstep=1cm,ystep=1cm] (-5,5) grid (5,5);
        \draw[->,color=black] (-5,0.) -- (5,0.); %
        \foreach \x in {-4.,-3.,-2.,-1.,1.,2.,3.,4}
        \draw[shift={(\x,0)},color=black] (0pt,2pt) -- (0pt,-2pt) node[below] {};
        \draw[->,color=black] (0.,-5) -- (0.,5); %
        \foreach \y in {-4.,-3.,-2.,-1.,1.,2.,3.,4}
        \draw[shift={(0,\y)},color=black] (1pt,0pt) -- (-1pt,0pt) node[left] {};
        \draw[color=black] (0pt,-10pt) node[right] {};
        \clip(-5,-5) rectangle (5,5);
        \draw [line width=0.5pt,domain=-8:8] plot(\x,{(--3.-1.*\x)/1.});
        \draw [line width=0.5pt,domain=-8:8] plot(\x,{(-3.-1.*\x)/1.});

        \draw[line width=0.5pt,color=green!50!black,fill=green!50!black,pattern=crosshatch,pattern color=green!50!black!25](-3,6)--(8.32,-5.32)--(8.22,5.94);
        \draw[line width=0.5pt,color=green!50!black,fill=green!50!black,fill opacity=0.25](-7.76,4.76)--(-8,-6.2)--(3.2,-6.2)--(-7.76,4.76);
        
        \begin{footnotesize}
        \draw [fill=red, color=red] (-1.,-1.) circle    (3pt);
        \draw [fill=blue, color=blue] (-3,0) circle     (3pt);
        \draw [fill=blue, color=blue] (0,-3) circle     (3pt);
        \draw [fill=blue, color=blue] (3,0) circle      (3pt);
        \draw [fill=blue, color=blue] (0,3) circle      (3pt);
        \draw[color=black] (-3,-3) node {$\calP_t^l(\alpha)$};
        \draw[color=black] (3,3) node {$\calP_t^u(\alpha)$};

        \draw[color=red] (-0.8,-1.4) node {$\alpha$};
        \draw[color=green!50!black] (-2.5,4)    node {$\beta_t^u(\alpha)$};
        \draw[color=green!50!black] (2.5,-4)    node {$\beta_t^l(\alpha)$};
        \draw[color=blue] (-3,1)        node {$-\frac{2\delta}{\vr_{t,1}}$};
        \draw[color=blue] (0.5,-2.5)    node {$ -\frac{2\delta}{\vr_{t,2}}$};
        \draw[color=blue] (2.9,-1.2)    node {$\frac{2\delta}{\vr_{t,1}}$};
        \draw[color=blue] (-0.75,2.3) node {$\frac{2\delta}{\vr_{t,2}}$};
    \end{footnotesize}
    \end{tikzpicture}
    \caption{Polytope partitioning for $d=2$. $\vr_{t,1}, \vr_{t,2}$ correspond to the $x_1$ and $x_2$ coordinates of $\vr_t(\alpha)$.\vspace{-25pt}}\label{fig:grinding_2d}
\end{wrapfigure}
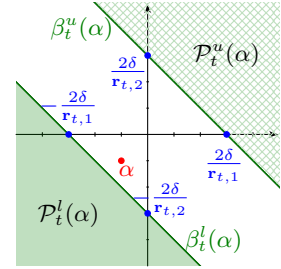

\paragraph{Algorithm Overview.} At each round $t$, \grinder (Algo.~\ref{algo:grinding}) maintains a sequence of nested \emph{polytopes} $\calP_t$, with $\calP_1 = \{\calA\}$ and decides which action $\alpha_t$ to play according to a two-stage sampling process. We denote the resulting distribution by $\calD_t$, and by $\Pr_t$ and $\f(\alpha)$ the associated probability and probability density function.

After the learner observes $\vr_t(\alpha_t)$, she computes two hyperplanes with the same normal vector $(\vr_t(\alpha_t))$ and symmetric intercepts ($\pm 4\sqrt{d}\delta$). These \emph{boundary} hyperplanes are defined as:  
\begin{align*}
\beta_t^u(\alpha_t)&:\forall \vw \in \calA, \langle \vw, \vr_t(\alpha_t) \rangle = 4 \sqrt{d} \delta\\
\beta_t^l(\alpha_t)&: \forall \vw \in \calA, \langle \vw, \vr_t(\alpha_t) \rangle = -4 \sqrt{d} \delta
\end{align*}
and they split the learner's action space into three \emph{regions}; one for which $\forall \mathbf{w}: \langle \mathbf{w}, \vr_t(\alpha_t) \rangle \geq 4\sqrt{d}\delta$, one for which $\langle \mathbf{w}, \vr_t(\alpha_t) \rangle \leq -4\sqrt{d}\delta$ and one for which $|\langle \mathbf{w}, \vr_t(\alpha_t) \rangle |\leq 4\sqrt{d}\delta$ (see Fig.~\ref{fig:grinding_2d}). %

Let {$H^{+}(\beta), H^{-}(\beta)$} denote the closed positive and negative halfspaces defined by hyperplane $\beta$ for intercept $4\sqrt{d}\delta$ and $-4\sqrt{d}\delta$ respectively\footnote{i.e., $\alpha \in H^+(\beta) \text{ if } \alpha \in \calA, \langle \beta, \alpha \rangle \geq 4\sqrt{d}\delta$ and similarly, $\alpha \in H^-(\beta) \text{ if } \alpha \in \calA, \langle \beta, \alpha \rangle \leq -4\sqrt{d}\delta$}. Slightly abusing notation, we say that polytope $p \subseteq H^+(\beta)$ if for all actions $\alpha$ contained in $p$ it holds that $\alpha \in H^+(\beta)$.

We define action $\alpha_t$'s \emph{upper} and \emph{lower} polytopes sets to be the sets of polytopes such that {$\calP_t^u(\alpha_t) = \{p \subseteq \calP_t, p \subseteq H^{+}(\beta_t^u(\alpha_t))\}$} and {$\calP_t^l(\alpha_t) = \{p \subseteq \calP_t, p \subseteq H^{-}(\beta_t^l(\alpha_t)) \}$} respectively. Defining these sets is useful since they represent the subsets of the learner's action space for which she can infer $\forall h: \ell(h, \vr_t(h),y_t)$ \emph{despite never observing $\vx_t$}! To be more precise for $h \in \calP_t^u(\alpha_t) \bigcup \calP_t^l(\alpha_t)$: 
\[\ell(h,\vr_t(h),y_t) = \1\{y_t=-1\} \cdot \1\left\{h\in \calP_t^u(\alpha_t) \right\} + \1\left\{y_t=+1 \right\} \cdot \1\left\{h \in \calP_t^l(\alpha_t) \right\}\]

The definition of the polytopes establishes that at each round the estimated loss within each polytope is \emph{constant}. If a polytope has \emph{not} been further ``grinded'' by the algorithm, then the estimated loss that was used to update the polytope has been the same within the actions of the polytope for each time step! This observation explains the way the weights of the polytopes are updated by scaling with the Lebesgue measure of each polytope. Due to the fact that the loss of all the points within a polytope is the same, we slightly abuse notation and we use $\ell(p, \vr_t(p), y_t)$ to denote the loss for any action $\alpha \in p$ for round $t$, if the agent best-responded to it. 

We next define the lower, upper, and middle $\sigma_t$-induced\footnote{These can only be computed if one has access to the agent's true datapoint $\sigma_t = (\vx_t,y_t)$. However, we only use them in our analysis, and $\grinder$ does not require access to them.} polytope sets as: $\calP_{t, \sigma_t}^l = \{ \alpha \in p, p \in \calP_t : \dist(\alpha, \vx_t) \leq -2\delta \}$, $\calP_{t, \sigma_t}^u = \{ \alpha \in p, p \in \calP_t : \dist(\alpha, \vx_t) \geq 2\delta\}$, and $\calP_{t, \sigma_t}^m = \{ \alpha \in p, p \in \calP_t : |\dist(\alpha, \vx_t)| < 2\delta \}$, where $\dist(\alpha,\vx_t) = \langle \alpha, \vx_t \rangle/\|\alpha\|_2$.

\grinder uses access to what we call an \emph{in-oracle} (Def.~\ref{def:in-oracle}). Our main regret theorem is stated for an accurate oracle, but we show that our regret guarantees still hold for approximation oracles (Lem.~\ref{lem:apx-oracle}). Such oracles can be constructed in practice, as we show in Sec.~\ref{sec:sims}.  

\begin{definition}[In-Oracle]\label{def:in-oracle}
We define the \emph{In-Oracle} as a black-box algorithm, which takes as input a polytope (resp. action) and returns the total \emph{in}-probability for this polytope (resp. action): 
\[\Prin\left[p\right] = \int_{\calA} \mathbb{P}_t \left[\left\{p \subseteq H^+\left(\beta_t^u\left(\alpha'\right)\right) \right\} \bigcup \left\{ p \subseteq H^-\left(\beta_t^l\left(\alpha'\right)\right) \right\} \right] d\alpha'
\]
\end{definition}
\begin{algorithm}[t!]
\caption{\grinder Algorithm for Strategic Classification}\label{algo:grinding}
\DontPrintSemicolon
\SetAlgoLined
Initialize polytopes' set: $\calP_0 = \{ \calA \}$. \;
Initialize polytope weights  $w_1(p) = \lambda(p), p \in \calP_0$. \;
Tune learning and exploration rates $\eta = \gamma \leq 1/2$, as specified in the analysis. \;
\For{$t \gets 1$ \KwTo $T$}{
Compute $\forall p \in \calP_t:\pi_t(p) = (1-\gamma)q_t(p) + \gamma \frac{\lambda(p)}{\lambda(\calA)}$. \tcp*{distribution over polytopes}%
\tcc{Two-stage sampling: first, polytope, second, draw action from within.}
Select polytope $p_t \sim \pi_t$ from which you draw action $\alpha_t \sim \texttt{Unif}(p_t)$ and commit to $\alpha_t$. \label{step:2stage-sample}\;
Observe the agent's response $\left(\vr_t(\alpha_t),y_t\right)$ to committed $\alpha_t$.\; 
\tcc{Space partitioning into smaller polytopes. $\calP_t:$ current polytopes set.}
Define a new set of polytopes $\calP_{t+1} = \calP_{t+1}^u(\alpha_t) \bigcup \calP_{t+1}^m(\alpha_t) \bigcup \calP_{t+1}^l(\alpha_t)$, where:\; 
\For{each polytope $p \in \calP_t$}{
    Add in $\calP_{t+1}^u(\alpha_t)$ the non-empty intersection $p \bigcap H^+(\beta_t^u(\alpha_t))$ \hspace{-0.1in}\tcp*{upper polytopes set}
    Add in $\calP_{t+1}^l(\alpha_t)$ the non-empty intersection $p \bigcap H^-(\beta_t^l(\alpha_t))$ \hspace{-0.7in}\tcp*{lower polytopes set}
    Add in $\calP_{t+1}^m(\alpha_t)$ the non-empty remainder of $p$. \tcp*{middle polytopes set}
}
    Compute $\hell(\alpha_t,\vr_t(\alpha_t),y_t) = \frac{\ell\left(\alpha_t,\vr_t(\alpha_t),y_t\right)}{\Prin[\alpha_t]}$.\tcp*{loss estimator for chosen action}
\For{each polytope $p \in \calP_{t+1}$}{
    \tcc{upper and lower polytopes get full information}
    Compute $\hell(p, \vr_t(p),y_t) = \frac{\ell(p, \vr_t(p),y_t) \cdot \1 \{p \subseteq \calP_{t+1}^u(\alpha_t) \bigcup \calP_{t+1}^l(\alpha_t) \}}{\Prin[p]}$. \;
    \tcc{weight scaling with the Lebesgue measure of the polytope}
    \label{step:mwu-update}Update $w_{t+1}(p) = \lambda(p) \exp \left(-\eta \sum_{\tau = 1}^{t} \hell(p, \vr_t(p), y_t)\right)$, $q_{t+1}(p) = \frac{w_{t+1}(p)}{\sum_{p' \in \calP_{t+1}}w_{t+1}(p')}$.\;
}
}
\end{algorithm}
We provide below the proof sketch for Thm~\ref{thm:regr-grind}. The full proof %
can be found in Appendix~\ref{app:grinding-algo}. We also note that the algorithm can be turned into one that does not assume knowledge of $T$ or $\lambda \left(\up\right)$ by using the standard \emph{doubling trick} \cite{auer1995gambling}.

\begin{proof}[Proof Sketch of Thm~\ref{thm:regr-grind}]%

Using properties of the pdf, we first prove that $\hell(\cdot, \cdot, \cdot)$ is an unbiased estimator, and its variance is inversely dependent on quantity $\Prin[\alpha]$. Next, using the $\sigma_t$-induced polytopes sets, we can bound the variance of our estimator (Lem.~\ref{lem:var}) by making a novel connection with a graph theoretic lemma from the literature in online learning with feedback graphs (\cite[Lem.~5]{NCBDK15}, also stated in Lem.~\ref{lem:graph-th}): \[\E_{\alpha_t \sim \calD_t} \left[\frac{1}{\Prin[\alpha_t]} \right] \leq 4\log\left(\frac{4\lambda\left(\calA\right)\cdot\abs*{\calP_{t, \sigma_t}^u \cup \calP_{t, \sigma_t}^l}}{\gamma \lambda \left(\up \right)}\right) + \lambda \left(\calP_{t, \sigma_t}^m \right)\]
To do so, we first expand the term $\E_{\alpha_t \sim \calD_t} \left[1/{\Prin[\alpha_t]} \right]$ as:   
\begin{align*}
\E_{\alpha_t \sim \calD_t} \left[\frac{1}{\Prin[\alpha_t]} \right] = \int_{\calA} \frac{\f(\alpha)}{\Prin[\alpha]} d\alpha = \underbrace{\int_{\bigcup \left(\calP_{t, \sigma_t}^u \cup \calP_{t, \sigma_t}^l\right)} \frac{\f(\alpha)}{\Prin[\alpha]} d\alpha}_{Q_1} + \underbrace{\int_{\bigcup\calP_{t, \sigma_t}^m} \frac{\f(\alpha)}{\Prin[\alpha]} d\alpha}_{Q_2} \numberthis{\label{eq:var}}
\end{align*}
Due to the fact that $\grinder$ uses \emph{conservative} estimates of the \emph{true} action space with $\dist(\alpha, \vx_t) \leq \delta$, term $Q_2$ can be upper bounded by $\lambda(\calP_{t,\sigma_t}^m)$. Upper bounding $Q_1$ is significantly more involved (Lem.~\ref{lem:var}). First, observe that each of the actions in $\calP_{t,\sigma_t}^u, \calP_{t,\sigma_t}^l$ gets updated with probability $1$ by any other action in the sets $\calP_{t,\sigma_t}^u, \calP_{t,\sigma_t}^l$. This is because for any of the actions in $\calP_{t,\sigma_t}^u, \calP_{t,\sigma_t}^l$, the agent could not have possibly misreported%
. So, for all actions $\alpha \in \calP_{t,\sigma_t}^u \cup \calP_{t,\sigma_t}^l$ we have that: $\Prin[\alpha] \geq \sum_{p \in \calP_{t,\sigma_t}^u \cup \calP_{t,\sigma_t}^l}\pi_t(p)$. As a result, we can instead think about the set of \emph{polytopes} that belong in $\calP_{t,\sigma_t}^u$ and $\calP_{t,\sigma_t}^l$ as forming a fully connected feedback graph. The latter, coupled with the fact that our exploration term makes sure that each polytope $p$ is chosen with probability at least ${\lambda(p)}/{\lambda(\calA)}$ gives that: $Q_1 \leq 4\log\left(\frac{4\lambda(\calA)\cdot\abs*{\calP_{t,\sigma_t}^u \cup \calP_{t,\sigma_t}^l}}{\lambda (\up) \cdot \gamma}\right)$. 

Expressing everything in terms of polytopes rather than individual actions is \emph{critical} in the previous step, because applying \cite[Lem.~5]{NCBDK15} on $\calA$, rather than $\calP_t$, gives vacuous regret upper bounds, due to the logarithmic dependence in the number of nodes of the feedback graph, which is infinite for the case of $\calA$. The penultimate step of the proof (Lem.~\ref{lem:sec-order}) is a second order regret bound for $\grinder$ on the estimated losses $\hell(\cdot, \cdot)$, which should be viewed as the continuous variant of the standard discrete-action second order regret bound for \texttt{EXP3}-type algorithms.     
In order to derive the bound stated in Thm~\ref{thm:regr-grind} we upper bound the total number of $\sigma_t$-induced polytopes with $\lambda(\calA)/\lambda(\up)$.
\end{proof}
The regret guarantee of \grinder is \emph{preserved} if instead of an \emph{accurate} in-oracle it is provided an \emph{$\eps$-approximate} one, where $\eps \leq 1/\sqrt{T}$ (Lemma~\ref{lem:apx-oracle}). As we also validate in Sec.~\ref{sec:sims}, in settings where few points violate the margin between the $+1$ and the $-1$ labeled points such approximation oracles do exist and are relatively easy to construct. 

Computing the volume of polytopes is a $\#$-P hard problem, so \grinder should be viewed as an information-theoretic result. However, if \grinder is provided access to an efficient black-box algorithm for computing the volume of a polytope, its runtime complexity is $\calO(T^d)$ (Lem.~\ref{lem:runtime}).

\subsection{Simulations}\label{sec:sims}

In this subsection we present our simulation results\footnote{Our code is publicly available here: \url{https://github.com/charapod/learn-strat-class}}. We build simulation datasets since in order to \emph{evaluate} the performance of our algorithms one needs to know the original datapoints $\vx_t$. The results of our simulations are presented in Fig.~\ref{fig:main-body}. 

\begin{figure}[t!]
\centering
\begin{subfigure}{0.33\textwidth}
    \includegraphics[scale=0.07]{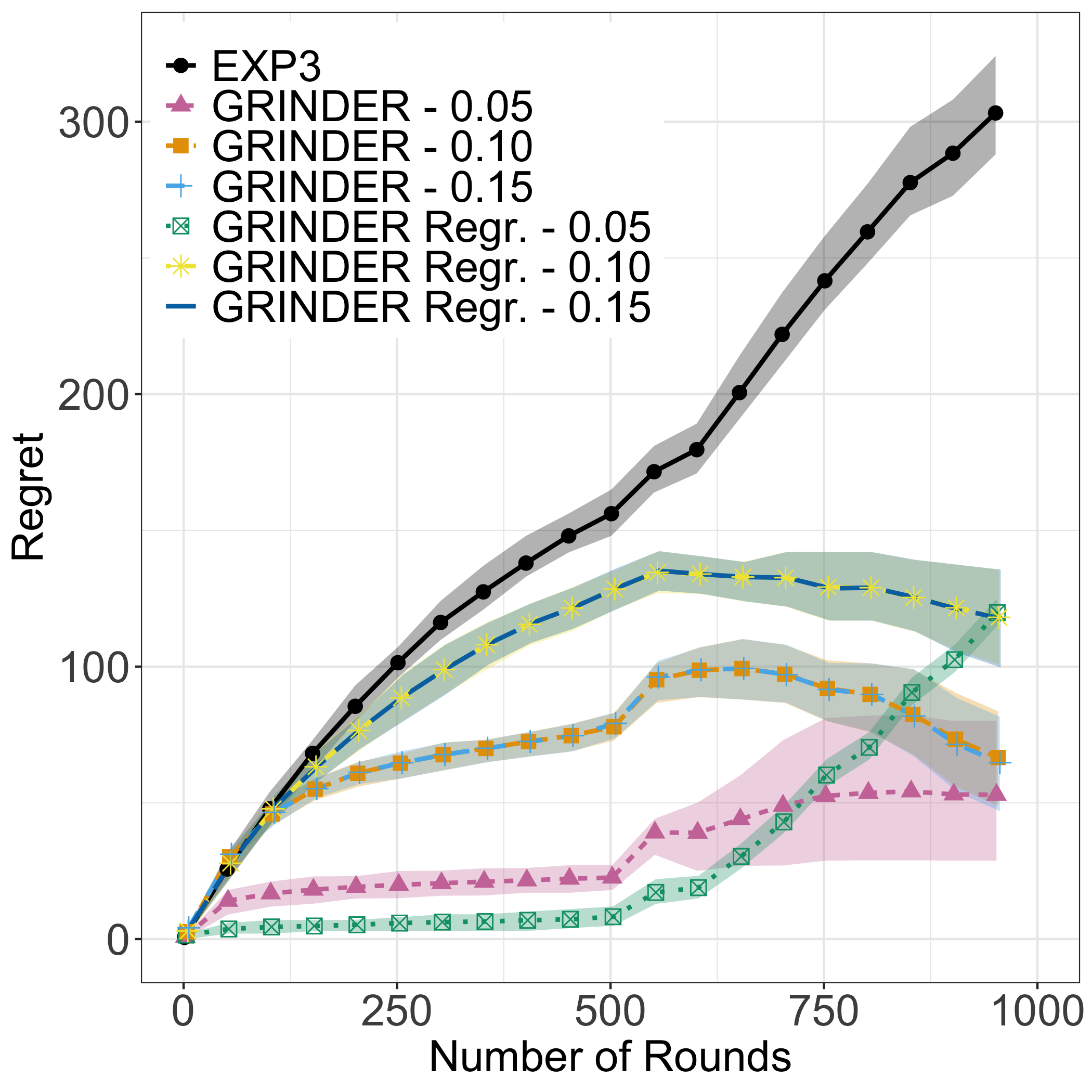}
\end{subfigure}%
\begin{subfigure}{0.33\textwidth}
    \centering
    \includegraphics[scale=0.07]{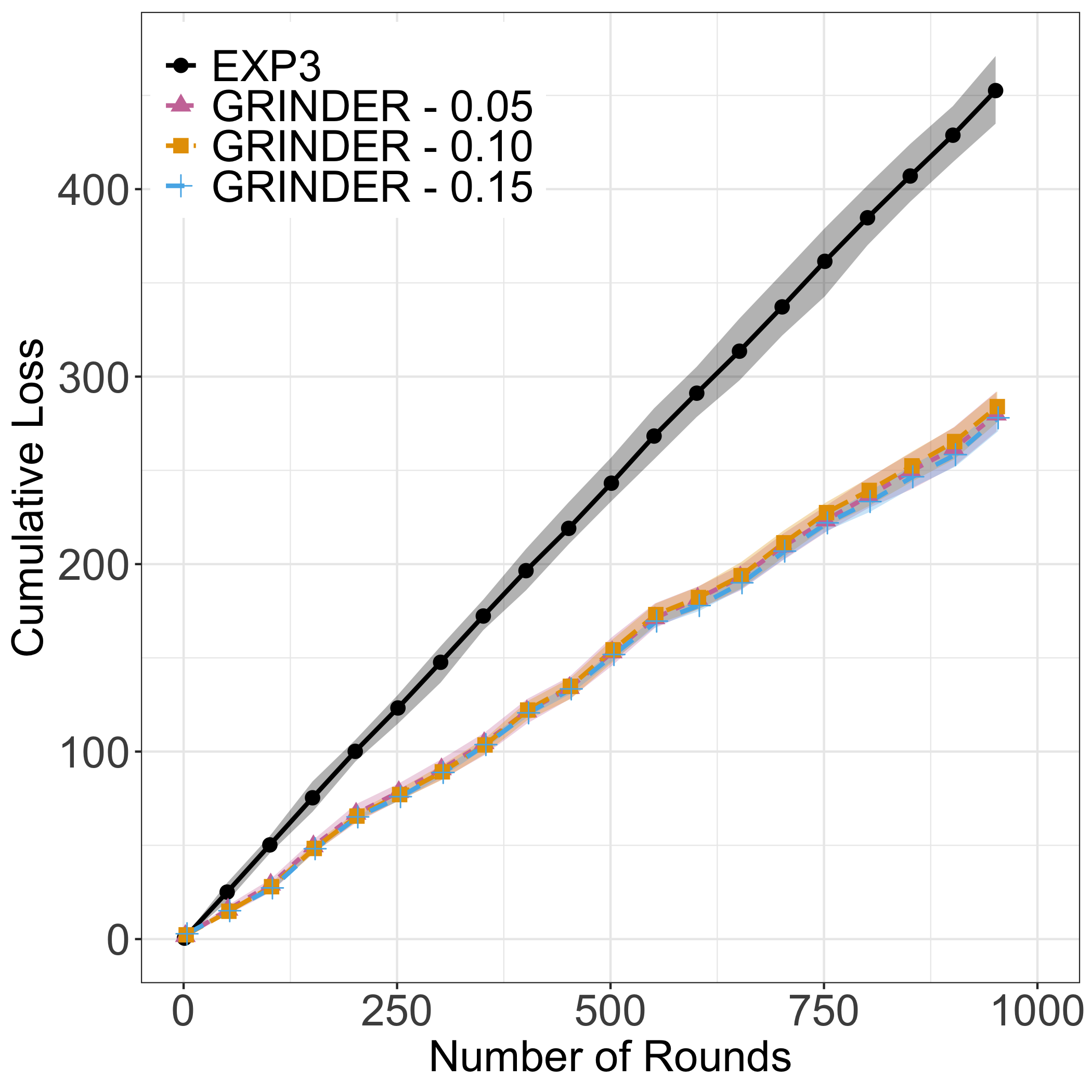}
\end{subfigure}%
\begin{subfigure}{0.33\textwidth}
    \centering
    \includegraphics[scale=0.07]{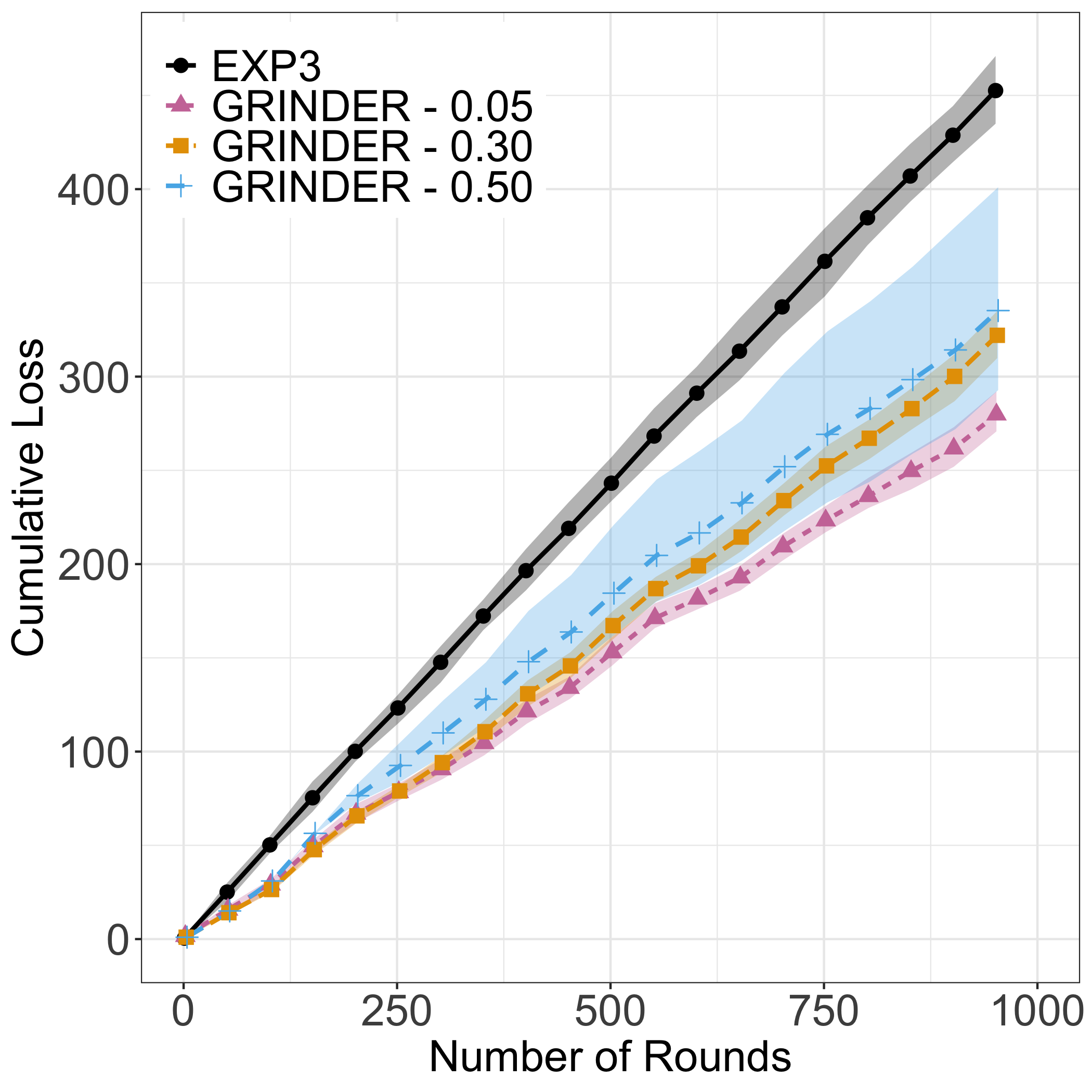}
\end{subfigure}%
\caption{Performance of \grinder vs. \texttt{EXP3}. In all cases, \grinder outperforms \texttt{EXP3}. Solid lines correspond to average regret/loss, and opaque bands correspond to $10$th and $90$th percentile. {\bf Left:} discrete action sets, accurate and regression oracle. {\bf Middle:} Continuous action set for \grinder with $\delta = 0.05, 0.10, 0.15$. {\bf Right:} Continuous action set for \grinder with $\delta = 0.05, 0.3, 0.5$. }\label{fig:main-body}
\end{figure}

For the simulation, we run \grinder against \texttt{EXP3} for a horizon $T=1000$, where each round was repeated for $30$ repetitions. The $\dBMR$ agents that we used are best-responding according to the utility function of Eq.~\eqref{eq:examp-dBMR}, and we studied $5$ different values for $\delta$: $0.05, 0.1, 0.15, 0.3, 0.5$. The $+1$ labeled points are drawn from Gaussian distribution as $\vx_t \sim (\calN(0.7,0.3), \calN(0.7,0.3))$ and the $-1$ labeled points are drawn from $\vx_t \sim (\calN(0.4, 0.3), \calN(0.4,0.3))$. Thus we establish that for the majority of the points there is a clear ``margin'' but there are few points that violate it (i.e., there exists no perfect linear classifier).

\texttt{EXP3} is always run with a fixed set of actions and always suffers a dependence on the different actions (i.e., not $\delta$). We then run \grinder in the same fixed set of actions and with a continuous action set. For the discrete action set, we include the results for both the accurate and the regression-based approximate oracle. We remark that if the action set is discrete, then $\grinder$ becomes similar to standard online learning with feedback graph algorithms (see e.g., \citep{NCBDK15}), but the feedback graph is built according to \dBMR agents. In this case, the regret scales as $\calO(a(G) \log T )$, where $a(G)$ is the independence number of graph $G$.

For the continuous action set it is not possible to identify the best-fixed action in hindsight. As a result, we report the cumulative loss. In Appendix~\ref{app:simulations}, we include additional simulations for a different family of $\dBMR$ agents\footnote{Namely, their utility function is: $u_t(\alpha_t, \vr_t(\alpha_t), \sigma_t) = \delta' \cdot \langle \alpha_t, \vr_t(\alpha_t) \rangle - \| \vx_t - \vr_t(\alpha_t) \|_2$.}, and different distributions of labels. 

In order to build the approximation oracle we used past data and we trained a logistic regression model for each polytope, learning the probability that it is updated. Our model has ``recency bias'' and gives more weight to more recent datapoints. We expect that for more accurate oracles, our results are strengthened, as proved by our theoretical bounds.

Validating our theoretical results, \grinder outperforms the benchmark, \emph{despite the fact that we use an approximation oracle}. We also see that in the discrete action set, where an accurate oracle can be constructed, \grinder performs much better than the regression oracle. As expected, \grinder's performance becomes worse as the power of the adversary increases (i.e., as $\delta$ grows larger).   

\begin{wrapfigure}[15]{R}{0.3\textwidth}
\centering
\includegraphics[scale=0.07]{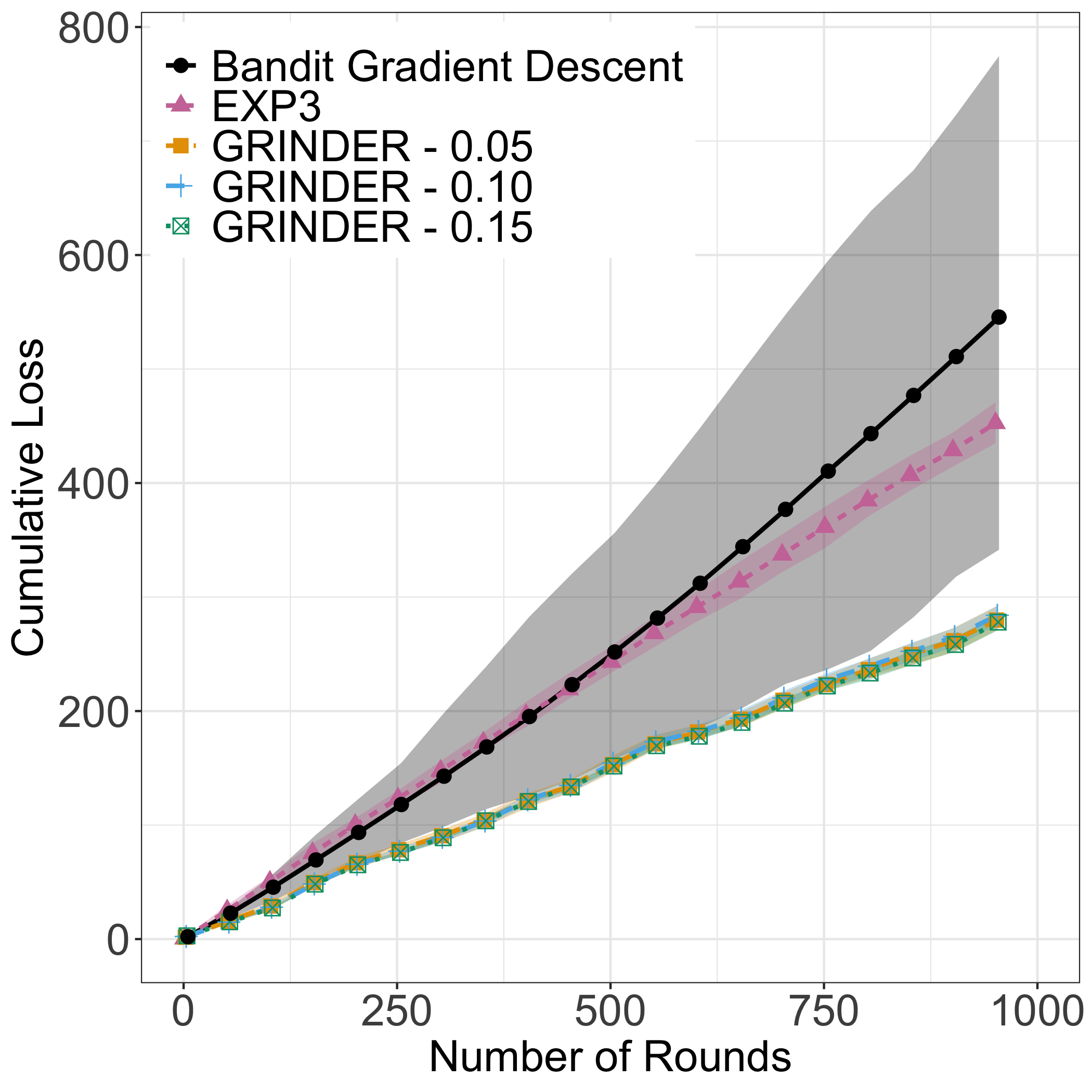}
\caption{Performance of \grinder vs. BGD.}
\label{fig:gd}
\end{wrapfigure}
\paragraph{Why not compare \grinder with a convex surrogate?} In this paragraph, we explain our decision to only compare \grinder with \texttt{EXP3}. In fact, \emph{no standard convex surrogate can be used} for learning linear classifiers against \dBMR, since the learner does not know precisely the agent's response \emph{function} $\vr_t(\alpha)$. As a result, the learner cannot guarantee that $\ell(\alpha, \vr_t(\alpha))$ is \emph{convex} in $\alpha$, even if $\ell(\alpha, \vz)$ ($\vz$ being independent of $\alpha$) is convex in $\alpha$! Concretely, think about the following counterexample: let $h = (1,1,-1), h' = (0.5, -1, 0.25)$ be two hyperplanes, a point $\vx = (0.55, 0.4), y = +1$, $\delta = 0.1$, and let $\ell(h, \vr(h)) = \max\{0, 1 - y \cdot \langle h, \vr(h) \rangle\}$ (i.e., hinge loss, which is convex). We show that when $(\vx, y)$ is a \dBMR agent, $\ell(\alpha, \vr(\alpha))$ is \emph{no longer convex} in $\alpha$. Take $b = 0.5$ and construct $h_b = 0.5h + 0.5h' = (0.75,0,-0.375) = (1, 0, -0.5)$. $(\vx,y)$ only misreports to (say) $(0.61, 0.4)$ when presented with $h$ (as $h_b$ and $h'$ classify $\vx$ as $+1$). Computing the loss: $\ell(h_b, \vr(h_b)) = 0.95, \ell(h, \vr(h)) = 0.99$ and $\ell(h', \vr(h')) = 0.875$, so, $\ell(h_b, \vr(h_b)) > b \ell(h, \vr(h)) + (1 - b) \ell(h', \vr(h'))$. Since in general $\ell(\alpha, \vr(\alpha))$ is not convex, it is may seem unfair to compare Bandit Gradient Descent (BGD) with \grinder but we include comparison results in Fig.~\ref{fig:gd}, where \grinder greatly outperforms BGD, for completeness. Identifying surrogate losses that are convex against \dBMR agents remains a very interesting open question.    

\section{Lower Bound}\label{sec:lb}

In this section we prove nearly matching lower bounds for learning a linear classifier against $\dBMR$ agents. To do so, we use the geometry of the sequence of datapoints $\sigma_t$ interpreted in the dual space. The proofs can be found in Appendix~\ref{app:lb}. %

\begin{theorem}\label{thm:lb}
For any strategy and any $\delta$, there exists a sequence of $\{\sigma_t\}_{t = 1}^T$ such that: 
\begin{equation}
\E \left[ \sum_{t \in [T]} \ell(\alpha_t, \vr_t(\alpha_t),y_t) \right] - \min_{\alphast \in \calA} \E \left[ \sum_{t \in [T]} \ell(\alphast, \vr_t(\alphast),y_t)\right] \geq \frac{1}{9\sqrt{2}} \sqrt{T \log \left( \frac{\lambda(\calA)}{\lambda(\tp)}\right)}
\end{equation}
where $\tp$ is the \emph{smallest $\sigma_t$-induced polytope from the sequence of $\{\sigma_t\}_{t=1}^T$}.
\end{theorem}

We remark here that that the $\sigma_t$-induced polytopes as defined in the previous section depend \emph{only} on the sequence of agents that we face, and not on the properties of any algorithm. We proceed with the proof sketch of our lower bound.

\vspace{-7pt}
\begin{proof}[Proof Sketch]
Fix a $\delta > 0$, and assume that the agents are \emph{truthful}\footnote{Truthful agents are \dBMR agents, so the lower bound holds for the whole family of \dBMR agents.} (i.e., $\vr_t(\alpha) = \vx_t, \forall t \in [T], \forall \alpha \in \calA$). Faithful to our model, however, the learner can only observe $\vr_t(\alpha)$, without knowing its equivalence to $\vx_t$. We prove the theorem in two steps.

In the first step (Lem.~\ref{lem:sqrt}) we show a more relaxed lower bound of order $\Omega(\sqrt{T})$. To prove this, we fix a particular feature vector $\vx$ for the agent, and two different adversarial environments (call them $U$ and $L$) choosing the label of $\vx$ according to different Bernoulli probability distributions; one of them favors label $y_t = +1$, while the other favors label $y_t = -1$. The $\Omega(\sqrt{T})$ lower bound corresponds to the regret accrued by the learner in order to distinguish between $U$ and $L$. 

For the second step, we separate the horizon into $\Phi = \log (\lambda(\calA)/\lambda(\tp))$ phases, each comprised by $T/\Phi$ consecutive rounds. In all rounds of a phase, the agent has the same $\vx$ and the labels are constructed by adversarial environments $U$ and $L$. At the end of each phase, either $U$ or $L$ must have caused regret at least $\Omega(\sqrt{T/\Phi})$. According to which one it was, nature selects the feature vector for the next phase in a way that guarantees that \emph{one} of the best-fixed actions for all previous phases is \emph{still part of the optimal actions} at this phase. The general pattern that we follow for the feature vectors of each phase is $\vx_\phi = \left(\frac{1}{2}, \frac{1}{4}\left(1+ \kappa_\phi \cdot 2^\phi\right) \right)$ where $\kappa_\phi$ is a phase-specific constant for which $\kappa_0 = 1$ and $\kappa_{\phi+1} = 2\kappa_\phi + 1$ if the environment causing regret $\Omega(\sqrt{T})$ was $U$ and $\kappa_{\phi+1} = 2\kappa_\phi - 1$ otherwise. This pattern establishes that the feature vectors are spaced in a way that every algorithm would be penalized enough, in order to be able to discern their labels.  
\end{proof}

\section*{Acknowledgments and Funding Disclosures}

The authors are grateful to Kira Goldner, Akshay Krishnamurthy, Thodoris Lykouris, David Parkes, and Vasilis Syrgkanis for helpful discussions at different stages of this work, and to the anonymous reviewers for their comments and suggestions.

This work was partially supported by the National Science Foundation under grants CCF-1718549 and IIS-2007951 and the Harvard Data Science Initiative.

\bibliographystyle{alpha}
\bibliography{refs,refs2}

\begin{thebibliography}{33}
\providecommand{\natexlab}[1]{#1}
\providecommand{\url}[1]{\texttt{#1}}
\expandafter\ifx\csname urlstyle\endcsname\relax
  \providecommand{\doi}[1]{doi: #1}\else
  \providecommand{\doi}{doi: \begingroup \urlstyle{rm}\Url}\fi

\bibitem[Alon et~al.(2015)Alon, Cesa-Bianchi, Dekel, and Koren]{NCBDK15}
Noga Alon, Nicolo Cesa-Bianchi, Ofer Dekel, and Tomer Koren.
\newblock Online learning with feedback graphs: Beyond bandits.
\newblock In \emph{JMLR WORKSHOP AND CONFERENCE PROCEEDINGS}, volume~40.
  Microtome Publishing, 2015.

\bibitem[Auer et~al.(1995)Auer, Cesa-Bianchi, Freund, and
  Schapire]{auer1995gambling}
Peter Auer, Nicolo Cesa-Bianchi, Yoav Freund, and Robert~E Schapire.
\newblock Gambling in a rigged casino: The adversarial multi-armed bandit
  problem.
\newblock In \emph{Proceedings of IEEE 36th Annual Foundations of Computer
  Science}, pages 322--331. IEEE, 1995.

\bibitem[Balcan et~al.(2015)Balcan, Blum, Haghtalab, and Procaccia]{BBHP15}
Maria-Florina Balcan, Avrim Blum, Nika Haghtalab, and Ariel~D Procaccia.
\newblock Commitment without regrets: Online learning in stackelberg security
  games.
\newblock In \emph{Proceedings of the sixteenth ACM conference on economics and
  computation}, pages 61--78. ACM, 2015.

\bibitem[Bechavod et~al.(2020)Bechavod, Ligett, Wu, and Ziani]{BLWZ20}
Yahav Bechavod, Katrina Ligett, Zhiwei~Steven Wu, and Juba Ziani.
\newblock Causal feature discovery through strategic modification.
\newblock \emph{arXiv preprint arXiv:2002.07024}, 2020.

\bibitem[Ben-Porat and Tennenholtz(2017)]{BPT17}
Omer Ben-Porat and Moshe Tennenholtz.
\newblock Best response regression.
\newblock In \emph{Advances in Neural Information Processing Systems}, pages
  1499--1508, 2017.

\bibitem[Ben-Porat and Tennenholtz(2018)]{BPT18}
Omer Ben-Porat and Moshe Tennenholtz.
\newblock Competing prediction algorithms.
\newblock \emph{arXiv preprint arXiv:1806.01703}, 2018.

\bibitem[Blum et~al.(2014)Blum, Haghtalab, and Procaccia]{BHP14}
Avrim Blum, Nika Haghtalab, and Ariel~D Procaccia.
\newblock Learning optimal commitment to overcome insecurity.
\newblock In \emph{Advances in Neural Information Processing Systems}, pages
  1826--1834, 2014.

\bibitem[Br{\"u}ckner and Scheffer(2011)]{BS11}
Michael Br{\"u}ckner and Tobias Scheffer.
\newblock Stackelberg games for adversarial prediction problems.
\newblock In \emph{Proceedings of the 17th ACM SIGKDD international conference
  on Knowledge discovery and data mining}, pages 547--555. ACM, 2011.

\bibitem[Bubeck et~al.(2011)Bubeck, Munos, Stoltz, and Szepesv{\'a}ri]{BMSS11}
S{\'e}bastien Bubeck, R{\'e}mi Munos, Gilles Stoltz, and Csaba Szepesv{\'a}ri.
\newblock X-armed bandits.
\newblock \emph{Journal of Machine Learning Research}, 12\penalty0
  (May):\penalty0 1655--1695, 2011.

\bibitem[Bubeck et~al.(2012)Bubeck, Cesa-Bianchi, et~al.]{BCB12}
S{\'e}bastien Bubeck, Nicolo Cesa-Bianchi, et~al.
\newblock Regret analysis of stochastic and nonstochastic multi-armed bandit
  problems.
\newblock \emph{Foundations and Trends{\textregistered} in Machine Learning},
  5\penalty0 (1):\penalty0 1--122, 2012.

\bibitem[Bubeck et~al.(2017)Bubeck, Lee, and Eldan]{BLE17}
S{\'e}bastien Bubeck, Yin~Tat Lee, and Ronen Eldan.
\newblock Kernel-based methods for bandit convex optimization.
\newblock In \emph{Proceedings of the 49th Annual ACM SIGACT Symposium on
  Theory of Computing}, pages 72--85. ACM, 2017.

\bibitem[Cai et~al.(2015)Cai, Daskalakis, and Papadimitriou]{CDP15}
Yang Cai, Constantinos Daskalakis, and Christos Papadimitriou.
\newblock Optimum statistical estimation with strategic data sources.
\newblock In \emph{Conference on Learning Theory}, pages 280--296, 2015.

\bibitem[Chen et~al.(2018)Chen, Podimata, Procaccia, and Shah]{CPPS18}
Yiling Chen, Chara Podimata, Ariel~D Procaccia, and Nisarg Shah.
\newblock Strategyproof linear regression in high dimensions.
\newblock In \emph{Proceedings of the 2018 ACM Conference on Economics and
  Computation}, pages 9--26. ACM, 2018.

\bibitem[Cohen et~al.(2016)Cohen, Hazan, and Koren]{CHK16}
Alon Cohen, Tamir Hazan, and Tomer Koren.
\newblock Online learning with feedback graphs without the graphs.
\newblock In \emph{International Conference on Machine Learning}, pages
  811--819, 2016.

\bibitem[Cummings et~al.(2015)Cummings, Ioannidis, and Ligett]{CIL15}
Rachel Cummings, Stratis Ioannidis, and Katrina Ligett.
\newblock Truthful linear regression.
\newblock In \emph{Conference on Learning Theory}, pages 448--483, 2015.

\bibitem[Dekel et~al.(2010)Dekel, Fischer, and Procaccia]{DFP10}
Ofer Dekel, Felix Fischer, and Ariel~D Procaccia.
\newblock Incentive compatible regression learning.
\newblock \emph{Journal of Computer and System Sciences}, 76\penalty0
  (8):\penalty0 759--777, 2010.

\bibitem[Dong et~al.(2018)Dong, Roth, Schutzman, Waggoner, and Wu]{drsww18}
Jinshuo Dong, Aaron Roth, Zachary Schutzman, Bo~Waggoner, and Zhiwei~Steven Wu.
\newblock Strategic classification from revealed preferences.
\newblock In \emph{Proceedings of the 2018 ACM Conference on Economics and
  Computation}, pages 55--70. ACM, 2018.

\bibitem[Flaxman et~al.(2005)Flaxman, Kalai, and McMahan]{FKM05}
Abraham~D Flaxman, Adam~Tauman Kalai, and H~Brendan McMahan.
\newblock Online convex optimization in the bandit setting: gradient descent
  without a gradient.
\newblock In \emph{Proceedings of the sixteenth annual ACM-SIAM symposium on
  Discrete algorithms}, pages 385--394. Society for Industrial and Applied
  Mathematics, 2005.

\bibitem[Hardt et~al.(2016)Hardt, Megiddo, Papadimitriou, and Wootters]{HMPW16}
Moritz Hardt, Nimrod Megiddo, Christos Papadimitriou, and Mary Wootters.
\newblock Strategic classification.
\newblock In \emph{Proceedings of the 2016 ACM conference on innovations in
  theoretical computer science}, pages 111--122. ACM, 2016.

\bibitem[Kleinberg et~al.(2008)Kleinberg, Slivkins, and Upfal]{KSU08}
Robert Kleinberg, Aleksandrs Slivkins, and Eli Upfal.
\newblock Multi-armed bandits in metric spaces.
\newblock In \emph{Proceedings of the fortieth annual ACM symposium on Theory
  of computing}, pages 681--690. ACM, 2008.

\bibitem[Lattimore and Szepesv{\'a}ri(2019)]{LS19}
Tor Lattimore and Csaba Szepesv{\'a}ri.
\newblock Bandit algorithms.
\newblock \url{https://tor-lattimore.com/downloads/book/book.pdf}, 2019.

\bibitem[Letchford et~al.(2009)Letchford, Conitzer, and Munagala]{LCM09}
Joshua Letchford, Vincent Conitzer, and Kamesh Munagala.
\newblock Learning and approximating the optimal strategy to commit to.
\newblock In \emph{International Symposium on Algorithmic Game Theory}, pages
  250--262. Springer, 2009.

\bibitem[Liu and Chawla(2009)]{WC09}
Wei Liu and Sanjay Chawla.
\newblock A game theoretical model for adversarial learning.
\newblock In \emph{2009 IEEE International Conference on Data Mining
  Workshops}, pages 25--30. IEEE, 2009.

\bibitem[Marecki et~al.(2012)Marecki, Tesauro, and Segal]{MTS12}
Janusz Marecki, Gerry Tesauro, and Richard Segal.
\newblock Playing repeated stackelberg games with unknown opponents.
\newblock In \emph{Proceedings of the 11th International Conference on
  Autonomous Agents and Multiagent Systems-Volume 2}, pages 821--828.
  International Foundation for Autonomous Agents and Multiagent Systems, 2012.

\bibitem[Meir et~al.(2011)Meir, Almagor, Michaely, and Rosenschein]{MAMR11}
Reshef Meir, Shaull Almagor, Assaf Michaely, and Jeffrey~S Rosenschein.
\newblock Tight bounds for strategyproof classification.
\newblock In \emph{The 10th International Conference on Autonomous Agents and
  Multiagent Systems-Volume 1}, pages 319--326. International Foundation for
  Autonomous Agents and Multiagent Systems, 2011.

\bibitem[Meir et~al.(2012)Meir, Procaccia, and Rosenschein]{MPR12}
Reshef Meir, Ariel~D Procaccia, and Jeffrey~S Rosenschein.
\newblock Algorithms for strategyproof classification.
\newblock \emph{Artificial Intelligence}, 186:\penalty0 123--156, 2012.

\bibitem[Perdomo et~al.(2020)Perdomo, Zrnic, Mendler-D\"{u}nner, and
  Hardt]{PZMH20}
Juan Perdomo, Tijana Zrnic, Celestine Mendler-D\"{u}nner, and Moritz Hardt.
\newblock Performative prediction.
\newblock In \emph{Proceedings of 37th International Conference on Machine
  Learning}, 2020.

\bibitem[Perote and Perote-Pena(2004)]{PP04}
Javier Perote and Juan Perote-Pena.
\newblock Strategy-proof estimators for simple regression.
\newblock \emph{Mathematical Social Sciences}, 47\penalty0 (2):\penalty0
  153--176, 2004.

\bibitem[Shavit et~al.(2020)Shavit, Edelman, and Axelrod]{SEA20}
Yonadav Shavit, Benjamin~L Edelman, and Brian Axelrod.
\newblock Causal strategic linear regression.
\newblock In \emph{Proceedings of the 37th International Conference on Machine
  Learning}, 2020.

\bibitem[Slivkins et~al.(2019)]{S19}
Aleksandrs Slivkins et~al.
\newblock Introduction to multi-armed bandits.
\newblock \emph{Foundations and Trends{\textregistered} in Machine Learning},
  12\penalty0 (1-2):\penalty0 1--286, 2019.

\bibitem[Stanley et~al.(2004)]{S04}
Richard~P Stanley et~al.
\newblock An introduction to hyperplane arrangements.
\newblock \emph{Geometric combinatorics}, 13:\penalty0 389--496, 2004.

\bibitem[Ustun et~al.(2019)Ustun, Spangher, and Liu]{ustun2019actionable}
Berk Ustun, Alexander Spangher, and Yang Liu.
\newblock Actionable recourse in linear classification.
\newblock In \emph{Proceedings of the Conference on Fairness, Accountability,
  and Transparency}, pages 10--19, 2019.

\bibitem[Zaslavsky(1975)]{Z75}
Thomas Zaslavsky.
\newblock Counting the faces of cut-up spaces.
\newblock \emph{Bulletin of the American Mathematical Society}, 81\penalty0
  (5):\penalty0 916--918, 1975.

\end{thebibliography}

\newpage
\appendix
{\Large {\bf Supplementary Material for Paper ID: 4406}}

\section{Appendix for Section~\ref{sec:regret-notions}}\label{app:prelims}

\subsection{Missing Proofs from Section~\ref{sec:regret-notions}}\label{app:regr2}

\begin{proof}[Proof of Theorem~\ref{thm:incomp-sr-ext}]
For completeness, we outline again the instance described earlier. Let an action space $\calA = \left\{h, h'\right\}$ such that $h = (1,1,-1)$ and $h' = (0.5, -1, 0.25)$, and let $\delta = 0.1$. The environment draws feature vectors $\vx^1 = (0.4, 0.5), \vx^2 = (0.6, 0.6), \vx^3 = (0.8, 0.9), \vx^4 = (0.65,0.3)$ with probabilities $p^1 = 0.05, p^2 = 0.15, p^3 = 0.05, p^4 = 0.75$ respectively, and with labels $y^1 = -1, y^2 = -1, y^3 = +1, y^4 = +1$. For clarity, Figure~\ref{fig:incomp-ext-stack} provides a pictorial depiction of the example, along with the best responses of the agents for each action. We first explain the values that the loss function takes according to the best-responses of the agents and the feature vectors drawn by nature.

\begin{itemize}
\item $\left(\E\left[ \ell \left(h, \vr_t\left( h \right), y_t\right) \right]\right)$ When the learner plays $h$ against agent's responses $\vr_t\left(h\right)$, she makes a mistake in her prediction every time that the environment drew $\vx_1$ or $\vx_2$ for the agent. This is because for $\vx^1$ the agent can misreport and fool the hyperplane. For $\vx^2$ the agent does not need to misreport; hyperplane $h$ classifies it as $+1$ erroneously already. For $\vx^4$ the agent can misreport and get correctly classified and for $\vx^3$ the hyperplane is correct all by itself. Hence: 
\[\E\left[ \ell \left(h, \vr_t\left( h \right) , y_t\right)\right] = \Pr\left[\text{nature draws } \vx^1 \text{ or } \vx^2\right]=p^1 + p^2 = 0.2\]
\item $\left(\E\left[ \ell \left(h', \vr_t\left( h' \right), y_t \right) \right]\right)$ When the learner plays $h'$ against agent's responses $\vr_t\left(h'\right)$, she makes a mistake in her prediction every time that the environment drew $\vx_1$ or $\vx_2$ or $\vx_3$ for the agent. This is because for both $\vx^1$ and $\vx^2$ the agent could misreport and fool the hyperplane and for $\vx^3$ the hyperplane classifies it incorrectly, but there is nothing that the learner can do to change it (due to $\delta$-boundedness). For $\vx^4$ the hyperplane classifies the point correctly, without the need of misreport from the agent. Hence: 
\[\E\left[ \ell \left(h', \vr_t\left( h' \right), y_t \right)\right] = \Pr\left[\text{nature draws } \vx^1 \text{ or } \vx^2 \text{ or } \vx^3\right] = p^1 + p^2 + p^3 = 0.25\]
\item $\left(\E\left[ \ell \left(h, \vr_t\left( h' \right), y_t\right) \right]\right)$ When the learner plays $h$ against agent's responses $\vr_t\left(h'\right)$, she makes a mistake in her prediction every time that the environment drew $\vx_2$ or $\vx_4$ for the agent, i.e., \[\E\left[ \ell \left(h, \vr_t\left( h' \right), y_t \right)\right]= \Pr\left[\text{nature draws } \vx^2 \text{ or } \vx^4\right] = p^2 + p^4 = 0.9\]
\item $\left(\E\left[ \ell \left(h', \vr_t\left( h \right), y_t \right) \right]\right)$ When the learner plays $h'$ against agent's responses $\vr_t\left(h\right)$, she makes a mistake in her prediction every time that the environment drew $\vx_3$ for the agent, i.e., \[\E\left[ \ell \left(h', \vr_t\left( h \right), y_t \right)\right]= \Pr\left[\text{nature draws } \vx^3 \right] = p^3 = 0.05\]
\end{itemize}

We now prove that \emph{any} sequence with sublinear Stackelberg regret will have linear external regret. Observe that for the Stackelberg regret, the best fixed action in hindsight is action $h$, with cumulative loss $0.2T$. Therefore, any action sequence that yields sublinear Stackelberg regret must have cumulative loss $0.2T + o(T)$, meaning that action $h'$ is played at most $o(T)$ times, while action $h$ is played at least $T - o(T)$ times. Given this, we proceed by identifying the best fixed action for the \emph{external} regret in any action such sequence $\left\{\alpha_t\right\}_{t=1}^T$. For that, we compute the loss that any of the actions in $\calA$ would incur, had they been the fixed action for sequence $\{\alpha_t\}_{t=1}^T$. 

Assume that action $h$ was the fixed action in hindsight for the sequence $\left\{\alpha_t\right\}_{t=1}^T$. Then, the cumulative loss incurred by playing $h$ constantly for $T$ rounds, denoted by $\sum_{t=1}^T \ell\left(h, \vr_t(\alpha_t), y_t\right)$ is:
\begin{align*}
\underbrace{0.2(T - o(T))}_{\substack{\text{loss incurred when playing} \\ \text{$h$ against $\vr_t(h)$}}}&+ \underbrace{0.9 o(T)}_{\substack{\text{loss incurred when playing} \\ \text{$h$ against $\vr_t(h')$}}} 
\end{align*}

Assume that action $h'$ was the fixed action in hindsight for the aforementioned action sequence. Then, the cumulative loss incurred by playing $h'$, denoted by $\sum_{t=1}^T \ell(h', \vr_t(\alpha_t), y_t)$ is equal to 
\begin{align*}
\underbrace{0.05(T - o(T))}_{\substack{\text{loss incurred when playing} \\ \text{$h'$ against $\vr_t(h)$}}}&+ \underbrace{0.25o(T)}_{\substack{\text{loss incurred when playing} \\ \text{$h'$ against $\vr_t(h')$}}} 
\end{align*}
Hence, we have that the best fixed action in hindsight for the external regret for the sequence $\left\{ \alpha_t \right\}_{t=1}^T$ is action $h'$. This means, however, that for the sequence $\left\{\alpha_t \right\}_{t=1}^T$, which guaranteed sublinear Stackelberg regret, the external regret is \emph{linear} in $T$: 
\begin{equation*}
R(T) \geq 0.2 T - 0.05 T \geq 0.15 T
\end{equation*}
Moving forward, we prove that \emph{any} action sequence with sublinear external regret will have linear Stackelberg regret. Since we previously proved that any action sequence $\left\{\alpha_t \right\}_{t=1}^T$ with sublinear Stackelberg regret plays at least $T - o(T)$ times action $h$ and this resulted in having linear external regret, we only need to consider sequences where action $h'$ is played $T - o(T)$ times, while action $h$ is played for $o(T)$ times. For any such action sequence, it suffices to show that the external regret will be sublinear, since for any such sequence the Stackelberg regret will be linear: 
\begin{equation*}
\calR(T) = 0.2o(T) + 0.25\cdot (T - o(T)) - 0.2 T \geq 0.05T
\end{equation*}
Similarly to the analysis above, we distinguish the following cases. Assume that action $h$ was the fixed action in hindsight for $\{\alpha_t\}_{t=1}^T$. Then, the cumulative loss incurred by playing $h$ is $\sum_{t=1}^T \ell\left(h, \vr_t(\alpha_t), y_t\right) = 0.2o(T) + 0.9(T - o(T))$. Assume that action $h'$ was the fixed action in hindsight for the aforementioned action sequence. Then, the cumulative loss incurred by playing $h'$ is $\sum_{t=1}^T \ell\left(h', \vr_t(\alpha_t), y_t\right) = 0.05o(T) + 0.25(T - o(T))$. As a result, the best fixed action in hindsight for the Stackelberg regret would be action $h'$, yielding external regret $o(T)$, i.e., sublinear. This concludes our proof. 
\end{proof}

\subsection{Purely Adversarial and Cooperative Stackelberg Games}\label{app:pure-sg}

Despite the worst-case incompatibility results that we have shown for the notions of external and Stackelberg regret, there are families of repeated games for which there is a clear \emph{hierarchy} between the two. In this subsection, we study two of the most important ones; the family of \emph{Purely Adversarial}, and the family of \emph{Purely Cooperative Stackelberg Games.}

\begin{definition}[Purely Adversarial Stackelberg Game (PASGs)]
We call a Stackelberg Game \emph{Purely Adversarial}, if for all actions $\alpha' \in \calA$ for the loss of the learner it holds that: $\ell(\alpha, \vr(\alpha), y_t) \geq \ell (\alpha, \vr(\alpha'), y_t)$, i.e., the agent inflicts the \emph{highest} loss to the learner, when best-responding to the action to which she committed.
\end{definition}

\begin{definition}[Purely Cooperative Stackelberg Game (PCSGs)]
We call a Stackelberg Game \emph{Purely Cooperative} if for all actions $\alpha' \in \calA$ for the loss of the learner it holds that: $\ell(\alpha, \vr(\alpha), y_t) \leq \ell (\alpha, \vr(\alpha'), y_t)$, i.e., the agent inflicts the \emph{lowest} loss to the learner, when best-responding to the action to which she committed.
\end{definition}

\begin{wraptable}{r}{0.3\textwidth}
    \centering
    \begin{tabular}{c|cc}
        {} &  $\vr_t(\alpha)$  & $\vr_t(\alpha')$ \\ \midrule
        {$\alpha$}  & $(7,-1)$ & $(6,-3)$ \\
        $\alpha'$ & $(6,-3)$ & $(7,-1)$ \\
   \end{tabular}
   \caption{Example of a PASG that is not zero-sum.}\label{ex:psg-not-zs}
\end{wraptable}
We remark here that despite their similarities, PASGs and PCSGs are \emph{not} equivalent to zero-sum games; in fact, it is easy to see that every zero-sum game is either a PASG or a PCSG, but the converse is \emph{not} true (see e.g., the example loss matrix given in Table~\ref{ex:psg-not-zs} where the first coordinate of tuple $(i,j)$ corresponds to the loss of the learner, and the second to the loss of the agent). Next, we outline the hierarchy between external and Stackelberg regret in repeated PASGs and PCSGs.
\begin{lemma}\label{lem:pasg3}
In repeated PASGs, Stackelberg regret is upper bounded by external regret, i.e., $\calR(T) \leq R(T)$. In other words, any no-Stackelberg regret sequence of actions is also a no-external regret one. 
\end{lemma}

\begin{proof}
Let $\talpha = \arg\min_{\alpha \in \calA} \sum_{t=1}^T \ell(\alpha, \vr_t(\alpha_t),y_t)$ and $\alphast = \arg\min_{\alpha \in \calA} \sum_{t=1}^T \ell(\alpha, \vr_t(\alpha), y_t)$. Then: 
\begin{align*}
R(T)  &= \sum_{t=1}^T \ell(\alpha_t, \vr_t(\alpha_t),y_t) - \sum_{t=1}^T \ell(\talpha, \vr_t(\alpha_t),y_t) \tag{definition of external regret}\\
            &\geq \sum_{t=1}^T \ell(\alpha_t, \vr_t(\alpha_t),y_t) - \sum_{t=1}^T \ell(\alphast, \vr_t(\alpha_t),y_t) \tag{definition of $\talpha$}\\
            &\geq \sum_{t=1}^T \ell(\alpha_t, \vr_t(\alpha_t),y_t) - \sum_{t=1}^T \ell(\alphast, \vr_t(\alphast),y_t) \tag{$\ell(\alphast, \vr_t(\alpha_t), y_t) \leq \ell(\alphast, \vr_t(\alphast),y_t)$} \\ 
            &= \calR(T)
\end{align*}
\end{proof}
On the other hand, for PCSGs it holds that: 
\begin{lemma}\label{lem:pcsg3}
In repeated PCSGs, Stackelberg regret is lower bounded by external regret, i.e., $\calR(T) \geq R(T)$. In other words, any no-external regret sequence of actions is also a no-Stackelberg regret one. 
\end{lemma}

\begin{proof}
Let $\talpha = \arg\min_{\alpha \in \calA} \sum_{t=1}^T \ell(\alpha, \vr_t(\alpha_t),y_t)$ and $\alphast = \arg\min_{\alpha \in \calA} \sum_{t=1}^T \ell(\alpha, \vr_t(\alpha),y_t)$. Then: 
\begin{align*}
R(T)        &= \sum_{t=1}^T \ell(\alpha_t, \vr_t(\alpha_t),y_t) - \sum_{t=1}^T \ell(\talpha, \vr_t(\alpha_t),y_t) \tag{definition of external regret}\\
            &\leq \sum_{t=1}^T \ell(\alpha_t, \vr_t(\alpha_t),y_t) - \sum_{t=1}^T \ell(\talpha, \vr_t(\talpha),y_t) &\tag{definition of PCSGs} \\
            &\leq \sum_{t=1}^T \ell(\alpha_t, \vr_t(\alpha_t),y_t) - \sum_{t=1}^T \ell(\alphast, \vr_t(\alphast),y_t) &\tag{definition of $\alphast$}\\
            &= \calR(T) 
\end{align*}
\end{proof}

\subsection{The Function $\ell(\alpha, \vr_t(\alpha),y_t)$}\label{app:loss-func}

As we mentioned in the main body, the learner's loss function $\ell(\alpha, \vr_t(\alpha),y_t)$ is generally not Lipschitz in her chosen action $\alpha$. For that, we study below the quantity $\left|\ell(\alpha, \vr_t(\alpha),y_t) - \ell(\alpha',\vr_t(\alpha'),y_t) \right|$.

\begin{lemma}
Let $\ell(x,y,z)$ denote the learner's loss function in a Stackelberg game, such that $\ell$ is $L_1$-Lipschitz with respect to the first argument, and $L_2$-Lipschitz with respect to the second. Then, for the learner's loss between any two actions $\alpha, \alpha' \in \calA$ it holds that: $$\abs*{\ell(\alpha,\vr_t(\alpha),y_t) - \ell(\alpha', \vr_t(\alpha'),y_t)} \leq \max \left\{L_1\cdot \norm*{\alpha' - \alpha}, L_2 \cdot \norm*{\vr_t(\alpha) - \vr_t(\alpha')}\right\}$$
\end{lemma}

\begin{proof}
We split the set of actions $\calA$ into pairs $(\alpha, \alpha')$ satisfying the following properties: 
\begin{enumerate}
\item\label{pair:1} For pair $(\alpha, \alpha')$ it holds that: $\ell(\alpha, \vr_t(\alpha),y_t) \geq \ell(\alpha, \vr_t(\alpha'),y_t)$ and $\ell(\alpha', \vr_t(\alpha'),y_t) \geq \ell(\alpha', \vr_t(\alpha),y_t)$. In other words, by best-responding the agent causes the biggest loss to the learner. Observe that, given that $\ell$ is $L_1$-Lipschitz in its first argument, we have that: 
\begin{equation*}
\ell(\alpha',\vr_t(\alpha'),y_t) - \ell(\alpha,\vr_t(\alpha),y_t) \geq \ell(\alpha', \vr_t(\alpha),y_t) - \ell(\alpha, \vr_t(\alpha),y_t) \geq -L_1 \norm*{\alpha' - \alpha}
\end{equation*}
and 
\begin{equation*}
\ell(\alpha',\vr_t(\alpha'),y_t) - \ell(\alpha, \vr_t(\alpha),y_t) \leq \ell(\alpha', \vr_t(\alpha'),y_t) - \ell(\alpha, \vr_t(\alpha'),y_t) \leq L_1 \norm*{\alpha'-\alpha}
\end{equation*}
Therefore, for such pairs of actions function $\ell(\alpha, \vr_t(\alpha),y_t)$ is $L_1$-Lipschitz with respect to $\alpha$. 
\item\label{pair:2} For pair $(\alpha, \alpha')$ it holds that: $\ell(\alpha, \vr_t(\alpha),y_t) \leq \ell(\alpha, \vr_t(\alpha'),y_t)$ and $\ell(\alpha', \vr_t(\alpha'),y_t) \leq \ell(\alpha', \vr_t(\alpha), y_t)$. In other words, by best-responding the agent causes the smallest loss to the learner. Similarly to Case~\ref{pair:1}, it is easy to see that on these pairs of actions, function $\ell(\alpha, \vr_t(\alpha),y_t)$ is again $L_1$-Lipschitz with respect to $\alpha$. 
\item \label{eq:not-lip-case3} For pair $(\alpha, \alpha')$ it holds that 
\begin{equation}\label{eq:lip1}
\ell(\alpha, \vr_t(\alpha),y_t) \geq \ell(\alpha, \vr_t(\alpha'),y_t)
\end{equation} 
and 
\begin{equation}\label{eq:lip2}
\ell(\alpha', \vr_t(\alpha'),y_t) \leq \ell(\alpha', \vr_t(\alpha), y_t)
\end{equation}
From Equations~\eqref{eq:lip1} and~\eqref{eq:lip2} we have that 
\begin{equation}\label{eq:lip-rhs-1}
\ell(\alpha', \vr_t(\alpha'),y_t) - \ell(\alpha, \vr_t(\alpha),y_t) \leq L_1 \norm*{\alpha' - \alpha}
\end{equation} We further distinguish the following cases: 
\begin{enumerate}
\item $\ell(\alpha, \vr_t(\alpha), y_t) = \ell(\alpha',\vr_t(\alpha'),y_t)$. Clearly, $\abs*{\ell(\alpha, \vr_t(\alpha),y_t) - \ell(\alpha', \vr_t(\alpha'),y_t)} \leq L_1 \cdot \norm*{\alpha' - \alpha}$ holds.
\item $\ell(\alpha, \vr_t(\alpha),y_t) \leq \ell(\alpha', \vr_t(\alpha'),y_t)$. From Equation~\eqref{eq:lip-rhs-1}, we get: $\abs*{\ell(\alpha, \vr_t(\alpha), y_t)- \ell(\alpha', \vr_t(\alpha'),y_t)} \leq L_1 \cdot \norm*{\alpha' - \alpha}$.
\item $\ell(\alpha, \vr_t(\alpha),y_t) \geq \ell(\alpha', \vr_t(\alpha'),y_t)$
Observe now that if $\ell(\alpha, \vr_t(\alpha),y_t) \geq \ell(\alpha', \vr_t(\alpha),y_t)$, then from Equation~\eqref{eq:lip2} the latter is lower bounded by $\ell(\alpha',\vr_t(\alpha'),y_t)$, which leads to a contradiction. Hence, it has to be the case that $\ell(\alpha, \vr_t(\alpha),y_t) \leq \ell(\alpha', \vr_t(\alpha),y_t)$. The latter, combined with the assumption that $\ell$ is $L_2$ - Lipschitz with respect to its second argument, implies that $\ell(\alpha', \vr_t(\alpha'),y_t) - \ell(\alpha, \vr_t(\alpha),y_t) \geq -L_2 \cdot \norm*{\vr_t(\alpha') - \vr_t(\alpha)}$. 
\end{enumerate}
\item For the pair $(\alpha, \alpha')$ it holds that $\ell(\alpha,\vr_t(\alpha),y_t) \leq \ell(\alpha,\vr_t(\alpha'),y_t)$ and $\ell(\alpha',\vr_t(\alpha'),y_t) \geq \ell(\alpha',\vr_t(\alpha),y_t)$. The case is analogous to Case~\ref{eq:not-lip-case3}.
\end{enumerate}
\end{proof}
To summarize, in PASGs (Case~\ref{pair:1} from aforementioned proof) and PCSGs (Case~\ref{pair:2} of aforementioned proof) the loss function written in terms of the action of the agent is \emph{Lipschitz}, i.e., $\abs*{\ell(\alpha, \vr_t(\alpha),y_t) - \ell(\alpha', \vr_t(\alpha'),y_t)} \leq L_1 \cdot \norm*{\alpha' - \alpha}$. However, in General Stackelberg Games one can only guarantee that 
\begin{equation}\label{eq:ub-loss}
\abs*{\ell(\alpha, \vr_t(\alpha),y_t) - \ell(\alpha', \vr_t(\alpha'),y_t)} \leq \max \left\{L_1 \cdot \norm*{\alpha' - \alpha}, L_2 \cdot \norm*{\vr_t(\alpha') - \vr_t(\alpha)} \right\}
\end{equation} 

Using Equation~\eqref{eq:ub-loss}, we show that there are some meaningful Stackelberg settings where $\norm*{\vr_t(\alpha') - \vr_t(\alpha)}$ can be upper bounded by $\norm*{\alpha' - \alpha}$ multiplied by a constant. For example, from well known results in convex optimization (for completeness see Lemma~\ref{lem:closeness-maxima}), we can see that this is exactly the case in settings where the agent's utility function, $u_t(\alpha, r)$ is \emph{strongly} concave in $r$, and quasilinear\footnote{Quasilinearity in $\alpha$ establishes that $L_{f,g}$ which is used by Lemma~\ref{lem:closeness-maxima} will be linear in $\norm*{\alpha' - \alpha}$.} in $\alpha$.

\begin{lemma}[Closeness of Maxima of Strongly Concave Functions (folklore)]\label{lem:closeness-maxima}
Let functions $f: \calX \mapsto \bbR, g: \calX \mapsto \bbR$ be two multidimensional, $1/\eta_c$-strongly concave functions with respect to some norm $|| \cdot ||$. Let $h(\vx) = f(\vx) - g(\vx), \vx \in \calX$ be $L_{f,g}$-Lipschitz\footnote{We use the subscript $f,g$ in the Lipschitzness constant to denote the fact that it depends on the two functions $f$ and $g$.} with respect to the same norm $||\cdot ||$. Then, for the maxima of the two functions: $\mu_f = \arg\max_{\vx \in \calX}f(\vx)$ and $\mu_g = \arg\max_{\vx \in \calX}g(\vx)$ it holds that: 
\begin{equation}
||\mu_f - \mu_g|| \leq L_{f,g}\cdot\eta_c
\end{equation}
\end{lemma}

\begin{proof}
First, we take the Taylor expansion of $f$ around its maximum, $\mu_f$ and use the strong concavity condition:
\begin{align}
f(\vx)  &\leq f(\mu_f) + \langle \nabla f(\mu_f), \vx - \mu_f \rangle - \frac{1}{2\eta}||\mu_f - \vx||^2 &\tag{strong concavity} \nonumber\\
        &= f(\mu_f) - \frac{1}{2\eta}||\mu_f - \vx||^2 &\tag{$\nabla f(\mu_f)=0$, since $\mu_f$ is the maximum} \label{eq:f-min} 
\end{align}
Similarly, by taking the Taylor expansion of $g$ around its maximum and using the strong concavity condition:
\begin{equation}\label{eq:g-min}
g(\vx) \leq g(\mu_g) - \frac{1}{2\eta}||\mu_g - \vx||^2 
\end{equation}
Using the $L_{f,g}$-Lipschitzness of $h(\vx)$ we get: 
\begin{align*}
L_{f,g} \cdot || \mu_g - \mu_f ||   &\geq |h(\mu_g) - h(\mu_f)| \geq h(\mu_g) - h(\mu_f) \\ 
                        &\geq f(\mu_g) - f(\mu_f) + g(\mu_f) - g(\mu_g) \\
                        &\geq \frac{1}{2\eta}||\mu_f - \mu_g||^2 + \frac{1}{2\eta}||\mu_f - \mu_g||^2 &\tag{from Taylor expansion} \\ 
                        &\geq \frac{1}{\eta}||\mu_f - \mu_g||^2
\end{align*}
Dividing both sides with $|| \mu_g - \mu_f ||$ concludes the proof.
\end{proof}

An example of such a utility function in the context of strategic classification (similar to the family of utility functions used in \citep{drsww18}) is presented below. 
\paragraph{Example.} Let $u_t(\alpha, \vr(\alpha), \sigma_t) = \langle \alpha, \vr(\alpha) \rangle - (\vx - \vr(\alpha))^2$. Then, we would like to compute an upper bound on the difference between $\left \|\vr({\alpha}) - \vr\left(\alpha'\right) \right\|$, where $\vr({\alpha}) = \arg \max_{\vz \in \calX; \vx} u_t(\alpha, \vz, \sigma_t)$ and $\vr({\alpha'}) = \arg \max_{\vz \in \calX;\vx} u_t(\alpha', \vz, \sigma_t)$. Following Lemma~\ref{lem:closeness-maxima} we can define functions $f(\vz) = u_t(\alpha, \vz, \sigma_t)$ and $g(\vz) = u_t(\alpha', \vz, \sigma_t)$. Now, observe that function $h(\vz) = f(\vz) - g(\vz)$ is indeed $\|\alpha - \alpha'\|$-Lipschitz (i.e., the Lipschitzness constant depends on the specific actions):  
\begin{equation*}
\left| f(\vy) - g(\vy) - f(\vz) + g(\vz) \right| = \left| \langle \alpha - \alpha', \vy - \vz \rangle \right| \leq \| \alpha - \alpha' \| \cdot \| \vy - \vz \|
\end{equation*}
where the last inequality comes from the Cauchy-Schwartz inequality. Furthermore, observe that both $f(\cdot)$ and $g(\cdot)$ are $\frac{1}{2}$-strongly concave. Therefore, from Lemma~\ref{lem:closeness-maxima} we get that: $$\|\vr({\alpha}) - \vr({\alpha'})\| \leq \frac{\|\alpha - \alpha' \|}{2}$$

\section{Appendix for Section~\ref{sec:grinding-algo}}\label{app:grinding-algo}

\subsection{Notation Reference Tables.}\label{app:notation-tables}

Our model and proof use a lot of notation. For easier reference, we summarize the notation used in our analysis in Tables~\ref{table:notation1} and~\ref{table:notation2}.
\begin{center}
\begin{table}[htbp]
\begin{tabular}{ll}
\toprule
{\bf Variable} & {\bf Description} \\
\midrule \midrule
$d \in \bbN$                        & dimension of the problem \\
$\calA \subseteq [-1, 1]^{d+1}$     & learner's action space \\ 
$\alpha_t \in \calA$                & learner's committed action for round $t$ \\ 
$\calX \subseteq \left([0,1]^d, 1\right)$      & agent's feature vector space \\
$\mathcal{Y} = \{-1, + 1\}$           & labels' space \\ 
$\vx_t \in \calX$                   & agent's feature vector, \emph{as chosen by nature} \\ 
$\sigma_t = (\vx_t, y_t), y_t \in \mathcal{Y}$ & agent's labeled datapoint, \emph{as chosen by nature} \\
$\vr_t(\alpha_t,\sigma_t) \in \calX$ (simplified to $\vr_t(\alpha_t)$)        & agent's \emph{reported} feature vector \\ 
$\hy_t \in \mathcal{Y}$                        & $\vr_t(\alpha_t)$'s label \\ 
$\ell(\alpha_t, \vr_t(\alpha_t),y_t)$ & learner's loss for action $\alpha_t$ against agent's report $\vr_t(\alpha_t)$ \\
$u_t(\alpha_t,\vr_t(\alpha_t), \sigma_t)$   & agent's utility for reporting $\vr_t(\alpha_t)$, when learner commits to $\alpha_t$ \\ 
$R(T)$                                      & learner's \emph{external} regret after $T$ rounds \\
$\calR(T)$                                  & learner's \emph{Stackelberg} regret after $T$ rounds \\ 
$\lambda(A)$                                & Lebesgue measure of measurable space $A$\\ 
\bottomrule
\end{tabular}
\caption{Model Notation Summary}\label{table:notation1}
\end{table}
\end{center}%

\begin{center}
\begin{table}[htbp]
\begin{tabular}{ll}
\toprule
{\bf Variable} & {\bf Description} \\
\midrule \midrule
$\calP_t$                                   & set of active polytopes at round $t$ \\
$\bcalP_t$                                    & set of active point-polytopes at round $t$ \\
$\calD_t$                                   & induced distribution from $2$-step sampling process \\
$\Pr_{t}, \f$                & cdf and pdf of $\calD_t$\\
$\beta_t^u(\alpha_t): \langle \vr_t(\alpha_t), \vw \rangle = 4\sqrt{d} \delta$     & upper boundary hyperplane \\
$\beta_t^l(\alpha_t): \langle \vr_t(\alpha_t),\vw \rangle = -4\sqrt{d} \delta$    & lower boundary hyperplane \\
$H^+\left( \beta_t^u(\alpha) \right)$       & $\alpha' \in H^+\left( \beta_t^u(\alpha) \right)$, if $\langle \vr_t(\alpha), \alpha' \rangle \geq 4\sqrt{d} \delta$ \\
$H^-\left( \beta_t^l(\alpha) \right)$       & $\alpha' \in H^-\left( \beta_t^l(\alpha) \right)$, if $\langle \vr_t(\alpha), \alpha' \rangle \leq -4\sqrt{d} \delta$ \\
$\calP_t^u(\alpha)$                         & upper polytopes set ($p \in \calP_t: p \in H^+\left(\beta_t^u(\alpha)\right)$)\\
$\calP_t^l(\alpha)$                         & lower polytopes set ($p \in \calP_t: p \in H^-\left(\beta_t^l(\alpha)\right))$\\
$\calP_t^m(\alpha)$                         & middle polytopes set ($p \in \calP_t: p \notin \calP_t \setminus \left( \calP_t^l \bigcup \calP_t^l \right)$\\
$\Prin[\alpha], \Prin[p]$ & in-probability for $\alpha$ and $p$ (see Definition~\ref{def:in-oracle})\\
$\up$                                       & polytope ($\notin \bcalP_t$) with smallest Lebesgue measure at round $T$ \\ 
\bottomrule
\end{tabular}
\caption{Notation Summary for Regret Analysis of $\grinder$.}\label{table:notation2}
\end{table}
\end{center}

\subsection{Proof of Theorem~\ref{thm:regr-grind}.}

The proof of Theorem~\ref{thm:regr-grind} follows from a sequence of lemmas and claims presented below. By convention, we call a single point a \emph{point-polytope}, and we denote the set of all point-polytopes by $\bcalP$.

\begin{proposition}\label{cl:1}
The two-stage sampling probability distribution $\calD_t$ is equivalent to a one-stage probability distribution of drawing directly \emph{an action} from density $d\pi_t(\cdot)$.  
\end{proposition}

\begin{proof}
The one-stage probability distribution that draws an action from $\pi_t$ is equivalent to choosing an action $\alpha \in \calA$ from probability \emph{density} function: $d\pi_t(\alpha) = (1-\gamma)dq_t(\alpha) + \frac{\gamma}{\lambda(\calA)}$. The two-stage probability is: $d\pi_{\calD_t}(\alpha) = \frac{1}{\lambda(p)} \left( (1-\gamma)q_t(p) + \frac{\gamma \lambda(p)}{\lambda(\calA)}\right)$. Since $q_t(p) = \lambda(p)dq_t(\alpha), \forall \alpha \in p$, we get the result. 
\end{proof}

Moving forward we analyze the first and the second moment of the loss $\hell(\alpha,\vr_t(\alpha),y_t)$ for each action $\alpha$, based on the induced probability distribution $\calD_t$, assuming oracle access to $\Prin[\alpha]$. 

\begin{lemma}[First Moment]\label{lem:unbiased}
The estimated loss $\hell(\alpha,\vr_t(\alpha),y_t)$ is an unbiased estimator of the true loss $\ell(\alpha,\vr_t(\alpha),y_t)$, when actions are drawn from the induced probability distribution $\calD_t$.
\end{lemma}

\begin{proof}
For all the actions $\alpha \in \calA$, given Proposition~\ref{cl:1}, it holds that: 
\begin{align*}
\E_{\alpha_t \sim \calD_t} \left[ \hell(\alpha,\vr_t(\alpha),y_t) \right] &= \int_{\calA} \f\left(\alpha'\right) \frac{\ell(\alpha,\vr_t(\alpha),y_t)\1\left\{\alpha \in N^{out}(\alpha')\right\}}{\Prin[\alpha]}d\alpha'= \ell(\alpha,\vr_t(\alpha),y_t)
\end{align*}
\end{proof}

\begin{lemma}[Second Moment]\label{lem:sec-mom}
For the second moment of the estimated loss $\hell(\alpha,\vr_t(\alpha),y_t)$ with respect to the induced probability distribution $\calD_t$ it holds that: 
\begin{equation*}
\E_{\alpha_t \sim \calD_t} \left[ \hell(\alpha,\vr_t(\alpha),y_t)^2 \right] = \frac{\ell(\alpha,\vr_t(\alpha),y_t)^2}{\Prin\left[\alpha\right]} \leq \frac{1}{\Prin\left[\alpha\right]}
\end{equation*}
\end{lemma}

\begin{proof}
For all the actions $\alpha \in \calA$, given Claim~\ref{cl:1}, it holds that: 
\begin{align*}
\E_{\alpha_t \sim \calD_t} \left[ \hell(\alpha,\vr_t(\alpha),y_t)^2 \right] &= \int_{\calA} \f\left(\alpha'\right) \frac{\ell(\alpha,\vr_t(\alpha),y_t)^2\1\left\{\alpha \in N^{out}(\alpha')\right\}}{\Prin[\alpha]^2}d\alpha'= \frac{\ell(\alpha,\vr_t(\alpha),y_t)^2}{\Prin\left[\alpha\right]}\leq \frac{1}{\Prin\left[\alpha\right]}
\end{align*}
\end{proof}

\begin{lemma}\label{lem:var}
Let $\up(t) = \arg\min_{p \in \calP_t \setminus \bcalP_t} \lambda(p)$ be the polytope with the smallest Lebesgue measure (excluding point-polytopes) after $t$ rounds. Then, the following inequality holds:
\begin{equation*}
\E_{\alpha_t \sim \calD_t} \left[\frac{1}{\Prin[\alpha_t]} \right] \leq 4\log\left(\frac{4\lambda\left(\calA\right)\cdot\abs*{\calP_{t, \sigma_t}^u \bigcup \calP_{t, \sigma_t}^l}}{\gamma \lambda\left(\up(t)\right)}\right) + \lambda \left(\calP_{t, \sigma_t}^m \right)
\end{equation*}
\end{lemma}

\begin{proof}
By definition, we expand the term: $\E_{\alpha \sim \calD_t} \left[ \frac{1}{\Prin[\alpha_t]}\right]$ as follows:

\begin{align*}
\E_{\alpha_t \sim \calD_t} \left[\frac{1}{\Prin[\alpha_t]} \right] &= \int_{\calA} \frac{\f(\alpha)}{\Prin[\alpha]} d\alpha \\
&= \underbrace{\int_{\bigcup \left(\calP_{t, \sigma_t}^u \cup \calP_{t, \sigma_t}^l\right)} \frac{\f(\alpha)}{\Prin[\alpha]} d\alpha}_{Q_1} + \underbrace{\int_{\bigcup\calP_{t, \sigma_t}^m} \frac{\f(\alpha)}{\Prin[\alpha]} d\alpha}_{Q_2} \numberthis{\label{eq:var}}
\end{align*}
where by integrating over $\bigcup \calP$ we denote the integral over all \emph{actions} that belong in some polytope from the set $\calP$. In the right hand side of Equation~\eqref{eq:var}, term $Q_2$ is relatively easier to analyze. Due to the conservative estimates of the \emph{true} middle space (i.e., the actions such that $\dist (\alpha, \vx_t) \leq \delta$), the set of polytopes $\calP_{t,\sigma_t}^m$ contains \emph{all} the actions that actually belong in the $\sigma_t$-induced middle space, plus some other actions for which the agent could not have misreported, due to their $\delta$-boundedness. Now, for all the actions that actually belong in the $\sigma_t$-induced middle space, it holds that they only get information (i.e., get updated) when they are chosen by the algorithm, while for the rest of the actions that have ended up in our middle space, they could have been updated by other actions as well. Thus, it holds that: 
\begin{equation*}
\forall \alpha \in \bigcup \calP_{t,\sigma_t}^m: \Prin[\alpha] \geq \f(\alpha)
\end{equation*}
As a result: 
\begin{equation}\label{eq:rhs2}
Q_2 = \int_{\bigcup \calP_{t,\sigma_t}^m} \frac{\f(\alpha)}{\Prin[\alpha]} d\alpha \leq \int_{\bigcup \calP_{t,\sigma_t}^m} \frac{\f(\alpha)}{\f(\alpha)} d\alpha = \lambda\left(\calP_{t,\sigma_t}^m\right)
\end{equation}
Moving forward, we turn our attention to term $Q_1$. Assume now that an action $\alpha$ belongs in a polytope $p_{\alpha}$. Then, there are (weakly) more actions that can potentially update action $\alpha$, than the whole polytope in which it belongs, $p_{\alpha}$; indeed, in order to update the polytope, one must make sure that every action within it is updateable. As a result, $\Prin[\alpha] \geq \Prin[p_{\alpha}]$. Using this in Equation~\eqref{eq:var} we get that the first term of the RHS of the variance is upper bounded by: 
\begin{equation}\label{eq:var2}
Q_1 \leq \sum_{p \in \calP_{t,\sigma_t}^u \cup \calP_{t,\sigma_t}^l} \int_{p} \frac{\f(\alpha)}{\Prin[p]} d\alpha 
\end{equation}
Further, let $\Prin[p]_{u,l}$ be the part of $\Prin[p]$ that depends only in the updates that stem from actions in either the upper or the lower polytopes sets. As such: $\Prin[p]_{u,l} \leq \Prin[p]$ and the term in Equation~\eqref{eq:var2} can be upper bounded by:
\begin{equation}\label{eq:var3}
Q_1 \leq \sum_{p \in \calP_{t,\sigma_t}^u \cup \calP_{t,\sigma_t}^l} \frac{1}{\Prin[p]_{u,l}} \int_{p} \f(\alpha) d\alpha 
\end{equation}
where we have also used the fact that we gain oracle access to quantity $\Prin\left[p \right]_{u,l}$ and therefore, we treat it as a constant in the integral. Observe now that the term $\int_{p} \f(\alpha) d\alpha$ corresponds to the total probability that the action $\alpha_t$, which is chosen from the induced probability distribution $\calD_t$, belongs to polytope $p$, i.e., it is equal to $\pi_t(p)$. Hence, the upper bound in Equation~\eqref{eq:var3} can be relaxed to: 
\begin{equation}\label{eq:var4}
Q_1 \leq \sum_{p \in \calP_{t,\sigma_t}^u \cup \calP_{t,\sigma_t}^l} \frac{\pi_t(p)}{\Prin[p]_{u,l}} 
\end{equation}
As we have explained before, $\pi_t(p) = 0$, for $p \in \bcalP_t$ and as a result, we can disregard point-polytopes from our consideration for the rest of this proof. We now upper bound this term by using the graph-theoretic lemma of \citet[Lemma~5]{NCBDK15}, which we provide below for completeness. 

\begin{lemma}[{{\cite[Lemma~5]{NCBDK15}}}]\label{lem:graph-th}
Let $G = (V,E)$ be a directed graph with $|V| = K$, in which each node $i \in V$ is assigned a positive weight $w_i$ lower bounded by a positive scalar $\eps \in (0, 1/2)$, i.e., $w_i \geq \eps, \forall i \in V$. If $\sum_{i \in V}w_i \leq 1$ then, denoting by $\alpha^G$ the independence number of $G$ we have that: $$\sum_{i \in V}\frac{w_i}{w_i + \sum_{j \in N^{in}(i)}w_j} \leq 4\alpha^G\frac{4K}{\alpha^G \eps}$$
\end{lemma}

Observe that all the actions within the $\sigma_t$-induced upper and the lower polytopes set form the following feedback graph: each node corresponds to a polytope from one of the sets $\calP_{t,\sigma_t}^u, \calP_{t,\sigma_t}^l$. So the total number of nodes is at most $\abs*{\calP_{t,\sigma_t}^u \cup \calP_{t,\sigma_t}^l}$, where by $|S|$ we denote the cardinality of a set $S$. Each edge $(i,j)$ corresponds to \emph{information passing} from node $i$ to node $j$, i.e., the directed edge $(i,j)$ exists when the loss for actions of polytope $j$ can be computed by just observing the loss for action from the polytope $i$. However, for each action belonging in a polytope among the $\sigma_t$-induced upper and lower polytopes sets, we know that the agent could not possibly misreport, due to him being myopically rational and $\delta$-bounded, and as a result, the loss for all the actions within the upper and the lower polytopes sets can be computed! As a result, the independence number of this feedback graph is $\alpha^G = 1$. Using the fact that each polytope $p$ is chosen with probability at least $\pi_t(p) \geq \gamma \frac{\lambda(p)}{\lambda\left(\calA\right)} \geq \gamma \frac{\lambda(\up(t))}{\lambda\left(\calA\right)}$ we can apply Lemma~\ref{lem:graph-th} for $\eps =\gamma \frac{\lambda(\up(t))}{\lambda\left(\calA\right)}$ and $\alpha^G = 1$ and obtain: 
\begin{equation*}  
Q_1 \leq 4\log\left(\frac{4\lambda\left(\calA\right)\cdot\abs*{\calP_{t,\sigma_t}^u \bigcup \calP_{t,\sigma_t}^l}}{\lambda \left(\up(t)\right) \cdot \gamma}\right)
\end{equation*}
Summing up the upper bounds for $Q_1$ and $Q_2$ we get:
\begin{equation*}  
\E_{\alpha_t \sim \calD_t} \left[\frac{1}{\Prin[\alpha_t]} \right] \leq 4\log\left(\frac{4\lambda\left(\calA\right)\cdot\abs*{\calP_{t,\sigma_t}^u \cup \calP_{t,\sigma_t}^l}}{\lambda \left(\up(t)\right) \cdot \gamma}\right) + \lambda \left(\calP_{t,\sigma_t}^m \right)
\end{equation*}
\end{proof}

\begin{lemma}[Second Order Regret Bound]\label{lem:sec-order}
Let $q_1, \dots, q_T$ be the probability distribution over the polytopes defined by in Step~\ref{step:mwu-update} of Algorithm~\ref{algo:grinding} for the estimated losses $\hell(\alpha,\vr_t(\alpha),y_t), t \in [T]$.  Then, the second order regret bound induced by $\grinder$ is: 
\begin{equation}\label{eq:sec-order-pol}
\sum_{t=1}^T \sum_{p \in \calP_{t+1}} q_t(p) \hell(p,\vr_t(p),y_t) - \sum_{t=1}^T \hell(\alphast, \vr_t(\alphast),y_t) \leq \frac{\eta}{2} \sum_{t=1}^T\sum_{p \in \calP_{t+1}}q_t(p)\hell(p,\vr_t(p),y_t)^2 +\frac{1}{\eta}\log \left( \frac{\lambda\left(\calA \right)}{\lambda(\underline{p})}\right) 
\end{equation}
where $\up$ is the polytope with the smallest Lebesgue measure in the finest partition of space $\calA$: $\up = \arg\min_{p \in \calP_T \setminus \bcalP_T} \lambda(p)$.
\end{lemma}

\begin{proof}
Let $W_t = \sum_{p \in \calP_t}w_t(p)$. We upper and lower bound the quantity $Q = \sum_{t=1}^T\log \frac{W_{t+1}}{W_t}$. For the lower bound: 
\begin{equation}\label{eq:before-lb}
Q = \sum_{t=1}^T \log \left(\frac{W_{t+1}}{W_t}\right) = \log \left(\frac{W_T}{W_1}\right)
\end{equation}
Observe now that in $t=1$ there only exists one polytope (the whole $[-1,1]^{d+1}$ space), with a total weight of $\lambda(\calA)$ and a probability of $1$. In other words, all the actions within this polytope have the same weight, which is equal to $1$ (uniformly weighted). As a result, $\log W_1 = \log \left(\sum_{p \in \calP_1} \int_{\calA} 1 d\alpha \right)= \log \left(\lambda\left(\calA\right) \right)$. For term $\log W_T$ we have: 
\begin{align*}
\log W_T &= \log \left(\sum_{p \in \calP_T}w_T(p)\right) = \log \left(\int_{\calA} w_T(\alpha)d\alpha \right)\\ 
&= \log \left(\sum_{p \in \calP_T \setminus \bcalP_T}\lambda(p)\exp\left(- \eta\sum_{t=1}^T \hell(p,\vr_t(p),y_t)\right)  +  \int_{\bigcup \bcalP_T} \exp\left(-\eta\sum_{t=1}^T \hell(\alpha,\vr_t(\alpha),y_t)\right) d\alpha\right)\numberthis{\label{eq:prev}}
\end{align*} 
where the last equality is due to the fact that not further grinded polytopes have maintained the \emph{same} estimated loss, $\hell$, for \emph{all} their containing points at each round $t$ and we denote by $\bcalP_T$ the set of point-polytopes contained in $\calP_T$. 

Since the horizon $T$ is finite, set $\bcalP_T$ is essentially a set of points, and it has a Lebesgue measure of $0$. Hence,
\begin{equation*}
\int_{\bigcup \bcalP_T} \exp\left(-\eta \sum_{t=1}^T \hell(\alpha,\vr_t(\alpha),y_t)\right) d\alpha = 0
\end{equation*}

Let $\alphast = \arg\min_{\alpha \in \calA}\sum_{t=1}^T \hell(\alpha,\vr_t(\alpha),y_t)$ (i.e., the best fixed action in hindsight among the \emph{all} actions after $T$ rounds, irrespective of whether it belongs to $\bigcup \bcalP_T$ or $\bigcup \calP_T \setminus \bcalP_T$) and $\up \in \calP_T \setminus \bcalP_T$ be the polytope with the smallest Lebesgue measure in $\calP_T \setminus \bcalP_T$ (i.e., excluding point-polytopes). Then, denoting by $p^* \in \calP_T$ the polytope where $\alphast$ belongs to, among the set of active polytopes $\calP_T$, Equation~\eqref{eq:prev} becomes can be lower bounded as follows: 
\begin{align*}
\log W_T &= \log \left(\sum_{p \in \calP_T \setminus \bcalP_T}\lambda(p)\exp\left(- \eta\sum_{t=1}^T \hell(p,\vr_t(p),y_t)\right) \right)\\
&\geq \log \left(\lambda(\up)\sum_{p \in \calP_T \setminus \bcalP_T}\exp\left(- \eta\sum_{t=1}^T \hell(p,\vr_t(p),y_t)\right) \right) &\tag{$\lambda(p) \geq \lambda(\up), \forall p \in \calP_T \setminus \bcalP_T$} \\
&\geq \log \left(\lambda(\up) \cdot \exp\left(- \eta\sum_{t=1}^T \hell \left(p^*,\vr_t\left(p^*\right),y_t \right)\right) \right) &\tag{$e^{-x} \geq 0, \forall x$} \\
&= \log \left(\lambda\left( \up \right) \cdot \exp \left(- \eta \sum_{t=1}^T\hell(\alphast, \vr_t(\alphast),y_t) \right)\right) = \log \left( \lambda\left(\up \right)\right) - \eta \sum_{t=1}^T \hell\left(\alphast, \vr_t\left(\alphast\right),y_t\right)\numberthis{\label{eq:WT-lb}} 
\end{align*}
As a result: 
\begin{equation}\label{eq:lb}
Q = \log W_T - \log W_1 \geq \log \left(\frac{\lambda \left( \up \right)}{\lambda(\calA) }\right) -  \eta\sum_{t=1}^T \hell(\alphast, \vr_t(\alphast),y_t)
\end{equation} 
We move on to the upper bound of $Q$ now. Upper bounding quantity $\log \frac{W_{t+1}}{W_t}$ we get: 
\begin{align*}
\log &\left( \frac{W_{t+1}}{W_t} \right) = \log \left(\frac{\int_{\calA} w_t(\alpha)\exp\left(-\eta \hell(\alpha,\vr_t(\alpha),y_t) \right)d\alpha}{W_t}\right) \\
&= \log \left(\int_{\calA}q_t(\alpha) \exp \left(-\eta \hell(\alpha, \vr_t(\alpha),y_t) \right) d\alpha \right) \\
&\leq \log \left( \int_{\calA} q_t(\alpha) \left(1 - \eta \hell(\alpha,\vr_t(\alpha),y_t) + \frac{\eta^2}{2} \hell(\alpha, \vr_t(\alpha),y_t)^2 \right) d\alpha \right) \tag{$e^{-x} \leq 1-x+\frac{x^2}{2}, x\in [0,1]$} \\
&\leq \log \left( 1 - \eta \int_{\calA} q_t(\alpha)\hell(\alpha,\vr_t(\alpha),y_t)d\alpha + \frac{\eta^2}{2} \int_{\calA} q_t(\alpha)\hell(\alpha,\vr_t(\alpha),y_t)^2 d\alpha \right) \tag{$\int_{\calA} q_t(\alpha) d\alpha = 1$} \\ 
&\leq - \eta \int_{\calA} q_t(\alpha)\hell(\alpha,\vr_t(\alpha),y_t)d\alpha + \frac{\eta^2}{2} \int_{\calA} q_t(\alpha)\hell(\alpha,\vr_t(\alpha),y_t)^2d\alpha \tag{$\log(1-x) \leq x, x \leq 0$}
\end{align*}
Summing up for the $T$ rounds the latter becomes: 
\begin{equation}\label{eq:up-bound}
\sum_{t=1}^T \log \left( \frac{W_{t+1}}{W_t} \right) \leq - \sum_{t =1}^T \eta \int_{\calA} q_t(\alpha)\hell(\alpha,\vr_t(\alpha),y_t)d\alpha + \sum_{t=1}^T\frac{\eta^2}{2} \int_{\calA}q_t(\alpha)\hell(\alpha,\vr_t(\alpha),y_t)^2d\alpha
\end{equation}
Combining the upper and lower bounds of Equations~\eqref{eq:lb} and~\eqref{eq:up-bound} we get that: 
\begin{equation*}
\sum_{t=1}^T \int_{\calA} q_t(\alpha) \hell(\alpha,\vr_t(\alpha),y_t)d\alpha - \sum_{t=1}^T \hell(\alphast, \vr_t(\alphast),y_t) \leq \frac{\eta}{2}\sum_{t=1}^T  \int_{\calA} q_t(\alpha)\hell(\alpha,\vr_t(\alpha),y_t)^2d\alpha +\frac{1}{\eta}\log \left( \frac{\lambda\left(\calA\right)}{\lambda\left(\up \right)}\right) 
\end{equation*}
\end{proof}

We are now ready for the proof of Theorem~\ref{thm:regr-grind}.

\begin{proof}[Proof of Theorem~\ref{thm:regr-grind}]
By taking the expectation with respect to distribution $\calD_t$ in Lemma~\ref{lem:sec-order} we get that: 
\begin{align*}
\sum_{t=1}^T \int_{\calA} q_t(\alpha) \E_{\calD_t} \left[\hell(\alpha,\vr_t(\alpha),y_t) \right]d\alpha &- \sum_{t=1}^T \E_{\calD_t} \left[\hell(\alphast, \vr_t(\alphast),y_t)\right] \leq \\
&\leq \frac{\eta}{2}\sum_{t=1}^T \int_{\calA} q_t(\alpha)\E_{\calD_t} \left[\hell(\alpha,\vr_t(\alpha),y_t)^2\right]d\alpha + \frac{1}{\eta}\log \left( \frac{\lambda\left(\calA \right)}{\lambda\left( \up \right)} \right)
\end{align*}
Combining Lemmas~\ref{lem:unbiased}, \ref{lem:sec-mom} with the latter we get: 
\begin{align*}
\sum_{t=1}^T \int_{\calA} &q_t(\alpha) \ell(\alpha,\vr_t(\alpha),y_t)d\alpha - \sum_{t=1}^T \ell(\alphast, \vr_t(\alphast),y_t) \\
&\leq \sum_{t=1}^T \frac{\eta}{2}\int_{\calA} \frac{q_t(\alpha) }{\Prin [\alpha]} d\alpha +\frac{1}{\eta}\log \left( \frac{\lambda\left(\calA \right)}{\lambda\left( \up \right)} \right) \\
&\leq \sum_{t=1}^T \eta \int_{\calA} \frac{\pi_t(\alpha) }{\Prin [\alpha]} d\alpha +\frac{1}{\eta}\log \left( \frac{\lambda\left(\calA \right)}{\lambda\left( \up \right)} \right) &\tag{$\pi_t(\alpha) \geq (1 - \gamma) q_t(\alpha)$ and $\gamma \leq \frac{1}{2}$} \\
&\leq \sum_{t=1}^T \eta \left( 4\log \left( \frac{4\lambda\left( \calA \right) \abs*{\calP_{t,\sigma_t}^u \bigcup \calP_{t,\sigma_t}^l}}{\gamma\cdot \lambda\left(\up(t)\right)}\right) + \lambda\left(\calP_{t,\sigma_t}^m \right)\right)+ \frac{1}{\eta}\log \left( \frac{\lambda\left(\calA \right)}{\lambda\left( \up \right)} \right) &\tag{Lemma~\ref{lem:var}}
\end{align*}
Using the fact that $\int_{\calA} \pi_t(\alpha)d\alpha \leq \int_{\calA} q_t(\alpha)d\alpha + \gamma$, the latter becomes:  
\begin{equation*}
\calR(T) \leq \gamma T + \eta \sum_{t=1}^T \left( 4\log \left( \frac{4\lambda\left( \calA \right) \abs*{\calP_{t,\sigma_t}^u \bigcup \calP_{t,\sigma_t}^l}}{\gamma \cdot \lambda\left(\up(t)\right)}\right) + \lambda\left(\calP_{t,\sigma_t}^m \right)\right)+ \frac{1}{\eta}\log \left( \frac{\lambda\left(\calA \right)}{\lambda\left( \up \right)} \right) 
\end{equation*}
Setting $\gamma = \eta$:
\begin{equation*}
\calR(T) \leq \eta \sum_{t=1}^T \left(1 + 4\log \left( \frac{4\lambda\left( \calA \right) \abs*{\calP_{t,\sigma_t}^u \bigcup \calP_{t,\sigma_t}^l}}{\eta \cdot \lambda\left(\up(t)\right)}\right) + \lambda\left(\calP_{t,\sigma_t}^m \right)  \right) +\frac{1}{\eta} \log \left( \frac{\lambda\left(\calA \right)}{\lambda\left( \up \right)} \right)
\end{equation*}
which can be relaxed to: 
\begin{align*}
&\calR(T) \leq \eta \sum_{t=1}^T \left(1 + 4\log \left( \frac{4\lambda\left( \calA \right) \abs*{\calP_{t,\sigma_t}^u \bigcup \calP_{t,\sigma_t}^l}}{\lambda\left(\up(t)\right)}\right) + \lambda\left(\calP_{t,\sigma_t}^m \right)  \right) +\frac{1}{\eta} \log \left( \frac{\lambda\left(\calA \right)}{\lambda\left( \up \right)} \right)\\
&\leq \eta \sum_{t=1}^T \left(1 + 4\log \left( \frac{4\lambda\left( \calA \right) \abs*{\calP_{t,\sigma_t}^u \bigcup \calP_{t,\sigma_t}^l} T}{\lambda\left(\up\right)}\right) + \lambda\left(\calP_{t,\sigma_t}^m \right)  \right) +\frac{1}{\eta} \log \left( \frac{\lambda\left(\calA \right)}{\lambda\left( \up \right)} \right) &\tag{$\lambda(\up(t)) \geq \lambda(\up)$} \\
&\leq \eta \cdot \max_{t \in [T]}\left\{1 + 4\log \left( \frac{4\lambda\left( \calA \right) \abs*{\calP_{t,\sigma_t}^u \bigcup \calP_{t,\sigma_t}^l} T}{\lambda\left(\up\right)}\right) + \lambda\left(\calP_{t,\sigma_t}^m \right)  \right\}\cdot T +\frac{1}{\eta} \log \left( \frac{\lambda\left(\calA \right)}{\lambda\left( \up \right)} \right) 
\end{align*}
Tuning $\eta$ to be $$\eta = \sqrt{\frac{\log \left( \frac{\lambda \left(\calA \right)}{\lambda\left(\up\right)}\right)}{\max_{t \in [T]}\left\{1 + 4\log \left( \frac{4\lambda\left( \calA \right) \abs*{\calP_{t,\sigma_t}^u \bigcup \calP_{t,\sigma_t}^l} T}{\lambda\left(\up\right)}\right) + \lambda\left(\calP_{t,\sigma_t}^m \right)  \right\}\cdot T}}$$ we get that the Stackelberg regret is upper bounded by: 
\begin{equation*}
\calR(T) \leq \calO \left( \sqrt{\max_{t \in [T]}\left\{\lambda \left(\calP_{t,\sigma_t}^m \right) + 4\log \left( \frac{4\lambda\left( \calA \right) \abs*{\calP_{t,\sigma_t}^u \bigcup \calP_{t,\sigma_t}^l}T}{\lambda(\up)} \right) +1 \right\} \cdot \log \left(\frac{\lambda\left(\calA\right)}{\lambda\left( \up \right)}\right) \cdot T }\right)
\end{equation*} 
Since the actions that belong in $\calP_{t,\sigma_t}^m$ are a subset of all the actions in $\calA$, then $\lambda \left(\calP_{t,\sigma_t}^m \right) \leq \lambda(\calA) = 1$. The set of all polytopes is upper bounded by $\frac{\lambda(\calA)}{\lambda(\up)}$ and hence, $\abs*{\calP_{t,\sigma_t}^u \bigcup \calP_{t,\sigma_t}^l} \leq \frac{\lambda(\calA)}{\lambda(\up)}$. Hence, for the Stackelberg regret we have: 
\begin{align*}
\calR(T) &\leq \calO \left(\sqrt{\log \left( \frac{\lambda(\calA)}{\lambda(\up)} T \right) \cdot \log \left(\frac{\lambda(\calA)}{\lambda(\up)} \right) T} \right) \\
&\leq \calO \left(\sqrt{\left(\lambda(\calA) + 1 + 4 \log \left( \frac{4\lambda(\calA)}{\lambda(\up)} \cdot \frac{\lambda(\calA)}{\lambda(\up)} \cdot T\right) \right) \cdot \log \left(\frac{\lambda(\calA)}{\lambda(\up)} \right) T }\right) \\ 
&\leq \calO \left(\sqrt{\left(\lambda(\calA) + 1 + 8 \log \left( \frac{2\lambda(\calA)}{\lambda(\up)} \cdot T\right) \right) \cdot \log \left(\frac{\lambda(\calA)}{\lambda(\up)} \right) T }\right) \\ 
&\leq \calO \left(\sqrt{\log \left( \frac{\lambda(\calA)}{\lambda(\up)} T\right) \cdot \log \left(\frac{\lambda(\calA)}{\lambda(\up)} \right) T }\right)  
\end{align*}
where the $\calO(\cdot)$ notation hides constants with respect to the horizon $T$.
\end{proof}

\subsection{Remaining Proofs}

\begin{lemma}\label{lem:apx-oracle}
For $\eps \leq 1/2$, we call {\upshape $\tPr_t[\alpha_t]$} an $\eps$-approximation oracle to $\Prin[\alpha_t]$, if $|\tPr_t [\alpha_t] - \Prin[\alpha_t] | \leq \eps \Prin[\alpha_t]$. Then, $\grinder$ run with oracle $\tPr_t[\cdot]$ instead of $\Prin[\alpha_t]$ achieves Stackelberg regret $\calR(T) \leq \calO(\sqrt{ T \log (T \lambda(\calA)/\lambda(\up))\cdot \log (\lambda(\calA)/\lambda(\up))}) + 2 \eps T$.   
\end{lemma}

\begin{proof}[Proof of Lemma~\ref{lem:apx-oracle}]
We start by computing how the first moment of estimator $\hell$ changes once you reweigh with $\tPr_t[\cdot]$ rather than $\Prin[\cdot]$: 
\begin{align*}
\E_{\alpha_t \sim \calD_t} \left[ \hell(\alpha,\vr_t(\alpha),y_t) \right] &= \int_{\calA} \f\left(\alpha'\right) \frac{\ell(\alpha,\vr_t(\alpha),y_t)\1\left\{\alpha \in N^{out}(\alpha')\right\}}{\tPr_t[\alpha]}d\alpha' \\ 
&= \ell(\alpha,\vr_t(\alpha),y_t) \cdot \frac{\Prin[\alpha]}{\tPr_t[\alpha]} \numberthis{\label{eq:bef-approx}}
\end{align*}
Since $\tPr_t[\alpha] \geq (1 - \eps) \Prin[\alpha]$, then from Equation~\eqref{eq:bef-approx} we have that: 
\begin{equation}\label{eq:1mom-apx-leq}
\E_{\alpha_t \sim \calD_t} \left[ \hell(\alpha,\vr_t(\alpha),y_t) \right]\leq \frac{\ell(\alpha,\vr_t(\alpha),y_t)}{1- \eps}
\end{equation}
Additionally, since $\tPr_t[\alpha] \leq (1 + \eps) \Prin [\alpha]$, then from Equation~\eqref{eq:bef-approx} we have that: 
\begin{equation}\label{eq:1mom-apx-geq}
\E_{\alpha_t \sim \calD_t} \left[ \hell(\alpha,\vr_t(\alpha),y_t) \right]\geq \frac{\ell(\alpha,\vr_t(\alpha),y_t)}{1 + \eps}
\end{equation}
We turn our attention to the second moment now, for which we will only need an upper bound.
\begin{align*}
\E_{\alpha_t \sim \calD_t} \left[ \hell(\alpha,\vr_t(\alpha),y_t)^2 \right] &= \int_{\calA} \f\left(\alpha'\right) \frac{\ell(\alpha,\vr_t(\alpha),y_t)^2\1\left\{\alpha \in N^{out}(\alpha')\right\}}{\tPr_t[\alpha]^2}d\alpha' = \frac{\ell(\alpha,\vr_t(\alpha),y_t)^2 \Prin\left[\alpha\right]}{\tPr_t\left[\alpha\right]^2} \\
&\leq \frac{1}{(1-\eps)^2 \Prin\left[\alpha\right]} \numberthis{\label{eq:2mom-apx-leq}}
\end{align*}
Lemma~\ref{lem:var} still holds without any change, as it is not affected by the exact definition of $\hell(\cdot)$, and so does Lemma~\ref{lem:sec-order}. Taking expectations in Lemma~\ref{lem:sec-order} we obtain the following: 
\begin{align*}
\sum_{t=1}^T \sum_{p \in \calP_{t+1}} q_t(p) \E \left[ \hell(p,\vr_t(p),y_t)\right] &- \sum_{t=1}^T \E \left[\hell(\alphast, \vr_t(\alphast),y_t)\right] \\
&\leq \frac{\eta}{2} \sum_{t=1}^T\sum_{p \in \calP_{t+1}}q_t(p)\E \left[ \hell(p,\vr_t(p),y_t)^2\right] +\frac{1}{\eta}\log \left( \frac{\lambda\left(\calA \right)}{\lambda(\underline{p})}\right) 
\end{align*}  
Applying Equations~\eqref{eq:1mom-apx-leq},~\eqref{eq:1mom-apx-geq} and~\eqref{eq:2mom-apx-leq} on the latter we obtain: 
\begin{align*}
\frac{1}{1 + \eps} \sum_{t=1}^T \int_{\calA} &q_t(\alpha) \ell(\alpha,\vr_t(\alpha),y_t)d\alpha - \frac{1}{1-\eps}\sum_{t=1}^T \ell(\alphast, \vr_t(\alphast),y_t) \\
&\leq \frac{\eta}{2} \frac{1}{(1-\eps)^2} \sum_{t=1}^T \int_\calA \frac{q_t(\alpha) d\alpha}{\Prin[\alpha]} +\frac{1}{\eta}\log \left( \frac{\lambda\left(\calA \right)}{\lambda(\underline{p})}\right)
\end{align*}
In the latter, applying Lemma~\ref{lem:var}, multiplying both sides by $1 - \eps$ and using the fact that $\eps \leq 1/2$ we obtain the result. 
\end{proof}

\begin{lemma}\label{lem:runtime}
Provided access to algorithms for computing the volume of a polytope and to an in-probability oracle, $\grinder$ has runtime complexity $\calO (T^d )$. 
\end{lemma}

\begin{proof}[Proof of Lemma~\ref{lem:runtime}]
With access to algorithms that compute the volume of a polytope and to an in-probability oracle, the complexity of $\grinder$ is dependent solely on the number of polytopes that get activated in the worst case. The latter depends on the number of new boundary hyperplanes that we introduce in the action space $\calA$ at each round.

If the sequence of real feature vectors $\{\vx_t\}_{t=1}^T$ is chosen adversarially, the number of new hyperplanes added in each round in $\calA$ is $2$. So, \emph{in the worst case}, after $T$ rounds we have $2T$ hyperplanes in general position in a $d$-dimensional space, which from \citet{Z75,S04} are: \[ \left|\calP_t\right| = O\left(\sum_{i=0}^d {2T \choose i} \right) = O\left(\frac{T^d}{d!}\right)\]
\end{proof}

\section{Appendix for Section~\ref{sec:lb}}\label{app:lb}

\begin{lemma}\label{lem:sqrt}
Fix a $\vr = \vx = (u)^d$, where by $(u)^d$ we denote the $d$-dimensional vector with $u \in [1/4,3/4]$ in every dimension. There exists a utility model for the agents, and a pair of adversarial environments $U$ and $L$ such that $\vr_t(\alpha) = \vx_t = \vx, \forall \alpha \in \calA, \forall t \in [T]$, and the sequence of $y_1, \dots, y_T$ is i.i.d. conditional on the choice of the adversary, such that: 
\begin{equation*}
\max_{\nu \in \{U,L\}} \min_{\alphast \in \calA} \E_{\nu} \left[\sum_{t \in [T]} \ell(\alpha_t, \vr_t(\alpha_t),y_t) - \sum_{t \in [T]} \ell(\alphast,\vr_t(\alphast),y_t) \right] \geq \frac{1}{9\sqrt{2}}\sqrt{T}
\end{equation*}
\end{lemma}

\begin{proof}
We are going to show this for the case where the agents $\forall t \in [T]$ are \emph{truthful}, i.e., they decide to report $\vr_t(\alpha) = \vx_t, \forall \alpha \in \calA, \forall t \in [T]$. Of course, the learner does not know (and cannot infer) that, so fix a $\delta > 0$ for the $\delta$-boundedness of the agents' utility function. We will prove the lemma only for deterministic strategies for the learner. As is customary, the claim for general strategies can be concluded by averaging over the learner's internal randomness and Fubini's theorem.

Fix an $\eps > 0$, and a scalar $u \in [1/4, 3/4]$, and define the adversarial environments as follows: $U$ is such that $y_t = +1$ with probability $1/2 + \eps$ and $y_t = -1$ with probability $1/2-\eps$, and $L$ is such that $y_t = -1$ with probability $1/2 + \eps$ and $y_t = +1$ with probability $1/2-\eps$. This means that under $U$, the majority of times the label is $+1$, and under $L$, the majority of times the label is $-1$. As a result, under $U$, \emph{any} action $\alpha$ such that $\langle \alpha, \vx \rangle \geq 2\delta$ is \emph{optimal} and under $L$, \emph{any} action $\alpha$ such that $\langle \alpha, \vx \rangle \leq - 2\delta$ is \emph{optimal}.

Take a sequence of actions $\alpha_1, \dots, \alpha_T$ and let $T_{\geq \delta}$ denote the \emph{number} of rounds for which $\langle \alpha_t, \vr \rangle \geq \sqrt{d} \delta$, and $T_{\leq -\delta}$ the number of rounds for which $\langle \alpha_t, \vr \rangle \leq -\sqrt{d}\delta$. Since $T_{\leq -\delta} + T_{\geq \delta} \leq T$ we get that: 
\begin{align*}
\E_{U} \left[ \calR(T) \right] &\geq \E_U \left[ \calR(T_{\leq -\delta})\right] \\
&\geq \sum_{t \in \left[ T_{\leq -\delta}\right]} \left[1 \cdot \left(\frac{1}{2}+ \eps \right) - 1 \cdot \left(\frac{1}{2} - \eps \right)\right] \\ 
&\geq 2\eps \E_U\left[ T_{\leq -\delta}\right] &\numberthis{\label{eq:first}}
\end{align*}
where the first inequality is due to the fact that $\ell(\alpha_t, \vx,y_t) = 0 = \ell(\alpha_U^*,\vx,y_t), \forall t \in \left[T_{\geq \delta} \right]$ and any optimal action $\alpha_U^*$ under $U$ as we reasoned before. The second inequality uses the following two facts: first, that $\ell(\alpha_U^*, \vx,y_t) = 1, \forall t \in \left[T_{\leq -\delta}\right]$, i.e., the best fixed action in hindsight when one encounters adversarial environment $U$ is an action that estimates the label of $\vx$ to be $1$. Second, that when playing against environment $U$, a learner incurs loss of $1$ every time that she predicted the label of $\vx$ to be $-1$ (which happens in at least all $T_{\leq -\delta}$ rounds), and the actual label was $1$ (which happens with probability $1/2 + \eps$). Similarly, we also see that 
\begin{equation}\label{eq:sec}
\E_L \left[ \calR(T)\right] \geq 2\eps \E_L \left[ T_{\geq \delta}\right]
\end{equation}
Let $\Pr_U, \Pr_L$ the distributions of $T_{\leq -\delta}, T_{\geq \delta}$ for adversarial environments $U,L$ respectively, and let $\Pr_m$ be the distribution of rounds when $y_t = +1$ with probability $1/2$. From Pinsker's inequality, and denoting by $\kl(p,q)$ the KL-divergence between distributions $p,q$, we have the following: 
\begin{equation}\label{eq:pinsker-U} 
\E_{U}\left[ T_{\leq - \delta} \right] \geq \E_{m} \left[T_{\leq -\delta} \right] - T \sqrt{\frac{\kl(\Pr_U, \Pr_m)}{2}} 
\end{equation}
and 
\begin{equation}\label{eq:pinsker-L} 
\E_{L}\left[ T_{\geq \delta} \right] \geq \E_{m} \left[T_{\geq \delta} \right] - T \sqrt{\frac{\kl(\Pr_U, \Pr_m)}{2}} 
\end{equation}
Then, from the data processing inequality for the KL-divergence we get:
\begin{equation}\label{eq:data-process-U}
\kl\left(\Pr_U, \Pr_m \right) \leq T \kl\left(\Bern\left( \frac{1}{2} + \eps \right), \Bern\left( \frac{1}{2} \right) \right) \leq 4T \eps^2  
\end{equation}
and 
\begin{equation}\label{eq:data-process-L}
\kl\left(\Pr_L, \Pr_m \right) \leq T \kl\left(\Bern\left( \frac{1}{2} + \eps \right), \Bern\left( \frac{1}{2} \right) \right) \leq 4T \eps^2  
\end{equation}
Plugging in Equations~\eqref{eq:data-process-U} and~\eqref{eq:data-process-L} in Equations~\eqref{eq:pinsker-U} and~\eqref{eq:pinsker-L} we get:
\begin{equation*}
\E_{U}\left[ T_{\leq - \delta} \right] \geq \E_{m} \left[T_{\leq -\delta} \right] - T \eps \sqrt{2T}
\end{equation*}
and 
\begin{equation*}
\E_{L}\left[ T_{\geq \delta} \right] \geq \E_{m} \left[T_{\geq \delta} \right] - T \eps \sqrt{2T}
\end{equation*}
Finally, averaging Equations~\eqref{eq:first} and~\eqref{eq:sec} and using the latter two Equations we get:
\begin{equation}
\max_{\nu \in \{U, L\}} \E_\nu \left[ \calR(T) \right) \geq \frac{\E_U \left[ \calR(T) \right] + \E_L \left[ \calR(T) \right]}{2} \geq \eps \left(T - 2\eps T \sqrt{2T} \right)
\end{equation}
Tuning $\eps = \frac{1}{3\sqrt{2T}}$ gives the result. 
\end{proof}

\begin{proof}[Proof of Theorem~\ref{thm:lb}]
We now assume without loss of generality that $2^\kappa = \lambda(\calA)/\lambda(\tp)$ for some constant $\kappa$. Let $\Phi$, such that $\Phi = \log \left( \frac{\lambda(\calA)}{\lambda(\tp)}\right)$, be a set of \emph{phases} created from $\frac{T}{\Phi}$ consecutive rounds in $T$. We create the following problem instance described for clarity in $2$ dimensions, and using truthful\footnote{This way we establish the creation of the $\sigma$-t induced polytopes from the way that we construct the sequence of datapoints.} agents.

First focus on phase $\phi = 1$ with feature vector $\vx_t = \vu = (1/2, 1/2)$, and specify the adversarial environments $U$ and $L$ exactly as in Lemma~\ref{lem:sqrt}. Then, after $T/\Phi$ rounds one of the two adversarial environments must have caused regret of at least $\sqrt{\frac{T}{162\Phi}}$ (Lemma~\ref{lem:sqrt}). If that environment was $U$, then it means that the majority of the labels for $\vu$ is $+1$ and hence, the best-fixed action in hindsight is $\alpha_U^*$ such that $\langle \alpha_U^*, \vu \rangle \geq 2\delta$. So fix the next phase's feature vector to be $\vx_t = \vu = (1/2, 5/8)$. Otherwise, if that environment was $L$, then it means that the majority of the labels for $\vu$ is $-1$ and hence, the best-fixed action in hindsight is $\alpha_L^*$ such that $\langle \alpha_L^*, \vu \rangle \geq 2\delta$ and you should fix the next feature vector to be $\vx_t = \vu = (1/2, 3/8)$. The reason for making this seemingly arbitrary choice is that we need to guarantee that one of the best-fixed actions for all previous phases, has survived and is still in the active set of actions in the current phase.

The general pattern that we follow is the following. At phase $\phi \in [\Phi]$, choose feature vector $\vx_\phi = \left(\frac{1}{2}, \frac{1}{4} \left(1 + \kappa_\phi \cdot 2^\phi \right) \right)$, where $\kappa_\phi$ is a phase-specific constant defined as follows. If at phase $\phi$, the adversarial environment incurring Stackelberg regret $\sqrt{\frac{T}{162\Phi}}$ was environment $U$, then, $\kappa_{\phi + 1} = 2 \kappa_\phi + 1$, else $\kappa_{\phi + 1} = 2\kappa_\phi - 1$. This is enough to establish that at \emph{any} phase $\phi$, there exists an action that would have been the best-fixed for \emph{all} previous phases despite which sequence of adversarial environments $U, L$ occurred. Another way to view this is similar to the polytope partitioning outlined in Figure~\ref{fig:grinding_2d}; presenting these feature vectors, we can guarantee that the polytope that held the best-fixed action so far, is still active. As a result, the regret for all $\Phi$ phases is equal to the sum of regrets of each phase. Additionally, the Lebesgue measure of the smallest polytope is $\up(\phi)$ and is a non-increasing function of the phases, even if it is announced to the learner that $\delta = 0$. As a result,
\begin{equation*}
\E \left[ \sum_{t \in [T]} \ell(\alpha_t, \vr_t(\alpha_t),y_t) \right] - \min_{\alphast \in \calA} \E \left[ \sum_{t \in [T]} \ell(\alphast, \vr_t(\alphast),y_t)\right] \geq \log \left(\frac{\lambda(\calA)}{\lambda(\tp)}\right)\frac{1}{9\sqrt{2}}\sqrt{\frac{T}{ \log \left( \frac{\lambda(\calA)}{\lambda(\tp)}\right)}}
\end{equation*}
\end{proof}

\section{Appendix for Section~\ref{sec:sims}}\label{app:simulations}

\subsection{Implementing $\grinder$ for Continuous Action Spaces}
In order to implement $\grinder$, we used the \texttt{polytope} library\footnote{\url{https://github.com/tulip-control/polytope/tree/master/requirements}}, which is part of the TuLiP python package. Other than some rounding-error fixes, we did not intervene with the core methods of the package. 

In order to implement the $2$-stage action draw method, we first chose a polytope (according to the probability function prescribed by $\grinder$) and then, by using \emph{rejection sampling} from the bounding box around the polytope, we chose the action associated with it. Note that this is equivalent to the theoretical $2$-stage draw.

In order to speed up our algorithm's performance, we also used the heuristic of bounding the allowable volume of any polytope to be greater than or equal to $0.01$, but in all the simulations that we tried, we saw comparable regret results even without the heuristic.  

\subsection{Logistic Regression Oracle}\label{app:oracle1}

In this subsection, we will outline our implementation of the logistic regression algorithm on the agents' past data, which serves as an estimate of the in-probability for each action. For ease of exposition, we provide the description of the oracle for the case of a predefined action set, and subsequently, we outline the way it generalizes to the continuous implementation.   

Before we embark on this, allow us first to observe that we already have a very crude (but potentially useful) lower bound for \emph{every} action $j \in \calA$. Indeed, each action \emph{always} updates itself, and actions that belong in the upper and lower polytope sets are always updated by all actions within these sets. The latter is due to the fact that for any hyperplane chosen within these sets, there is no possible manipulation from the perspective of the agent. We denote this crude lower bound for each action $j \in \calA$ by $c^j$. 

Labels are defined as $l_i^j = 0$ if action $j$ was \emph{not} updated at round $i$\footnote{In other words, action was at a distance less than $2\delta$ from the best-response of the round.}, and $1$ otherwise. As a first step, this oracle computes for each action $j \in \calA$ the probability that each action from $\calA$ updates $j$, by using a logistic regression\footnote{Technically, we run a different logistic regression for every action in $\calA$.} with feature vectors the set $H_{1:t}$, and $L_{1:t}^j$ as the labels. Let $p_i^j, i \in \calA$ correspond to the output probabilities, i.e., $p_i^j$ encodes the probability that action $j$ will be updated by action $i$. The in-probability of action $j$ is ultimately defined as: 
\begin{equation*}
\Pr^{in}[j] = \max \left\{ \sum_{i \in \calA} p_i^j \pi_t[i], c^j\right\}
\end{equation*} 

At a high-level, it is not hard to see how this can generalize to the continuous grinding case; instead of actions, one now uses whole \emph{polytopes}. The implementation, however, becomes significantly messier, as we need to propagate the history of past data for each polytope to its grinded sub-polytopes.

\subsection{Different Utility Function and Distribution of Datapoints}

The utility function that we assume for the agents at this subsection, is similar to the one studied by \citet{drsww18}, specifically: 
\begin{equation}\label{eq:ut2}
u_t(\alpha_t, \vr_t(\alpha_t),y_t) = \delta \cdot \langle \alpha_t, \vr_t(\alpha_t) \rangle - \|\vx_t - \vr_t(\alpha_t)\|_2 
\end{equation}
for values of $\delta = 0.05, 0.10, 0.15, 0.3, 0.5$. Similarly to the paper's main body, we run \grinder against \texttt{EXP3} for a horizon $T=1000$, where each round was repeated for $30$ repetitions. 

Fig.~\ref{fig:linear-distr1} presents the results for the case where the $+1$ labeled points are drawn as $\vx_t \sim (\calN(0.7,0.3), \calN(0.7,0.3))$ and the $-1$ labeled points are drawn from $\vx_t \sim (\calN(0.4, 0.3), \calN(0.4,0.3))$. The performance of \grinder compared to \texttt{EXP3} is similar to the one that we saw in Sec.~\ref{sec:sims} for the case of the different utility function. \grinder outperforms \texttt{EXP3}, and its performance degrades as the power of the agent (i.e., $\delta$) increases. We also see that in this case, the regression oracles are performing slightly worse that the regression oracles for the case of the utility function analyzed in Sec.~\ref{sec:sims}.

\begin{figure}[t!]
\centering
\begin{subfigure}{0.33\textwidth}
    \includegraphics[scale=0.07]{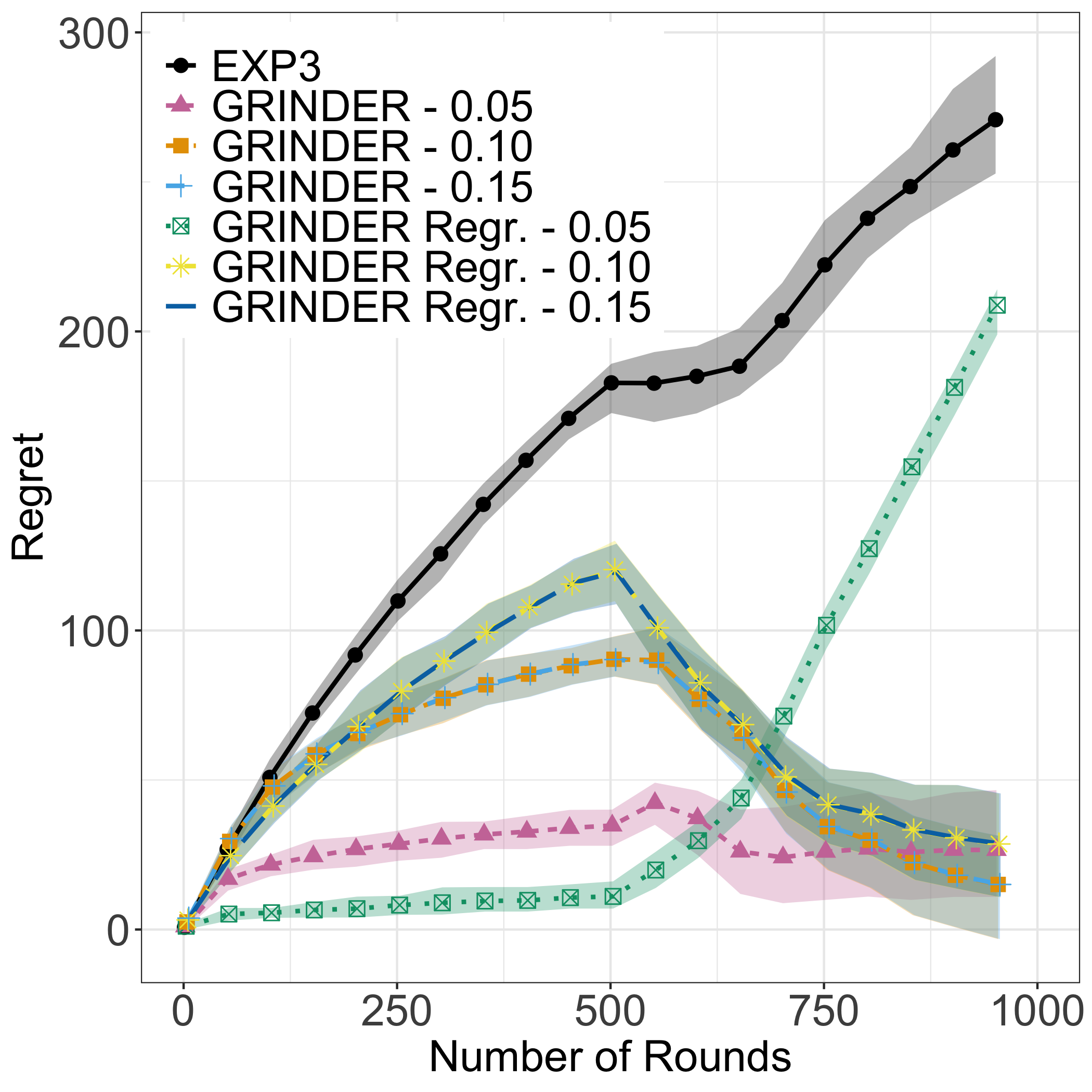}
\end{subfigure}%
\begin{subfigure}{0.33\textwidth}
    \centering
    \includegraphics[scale=0.07]{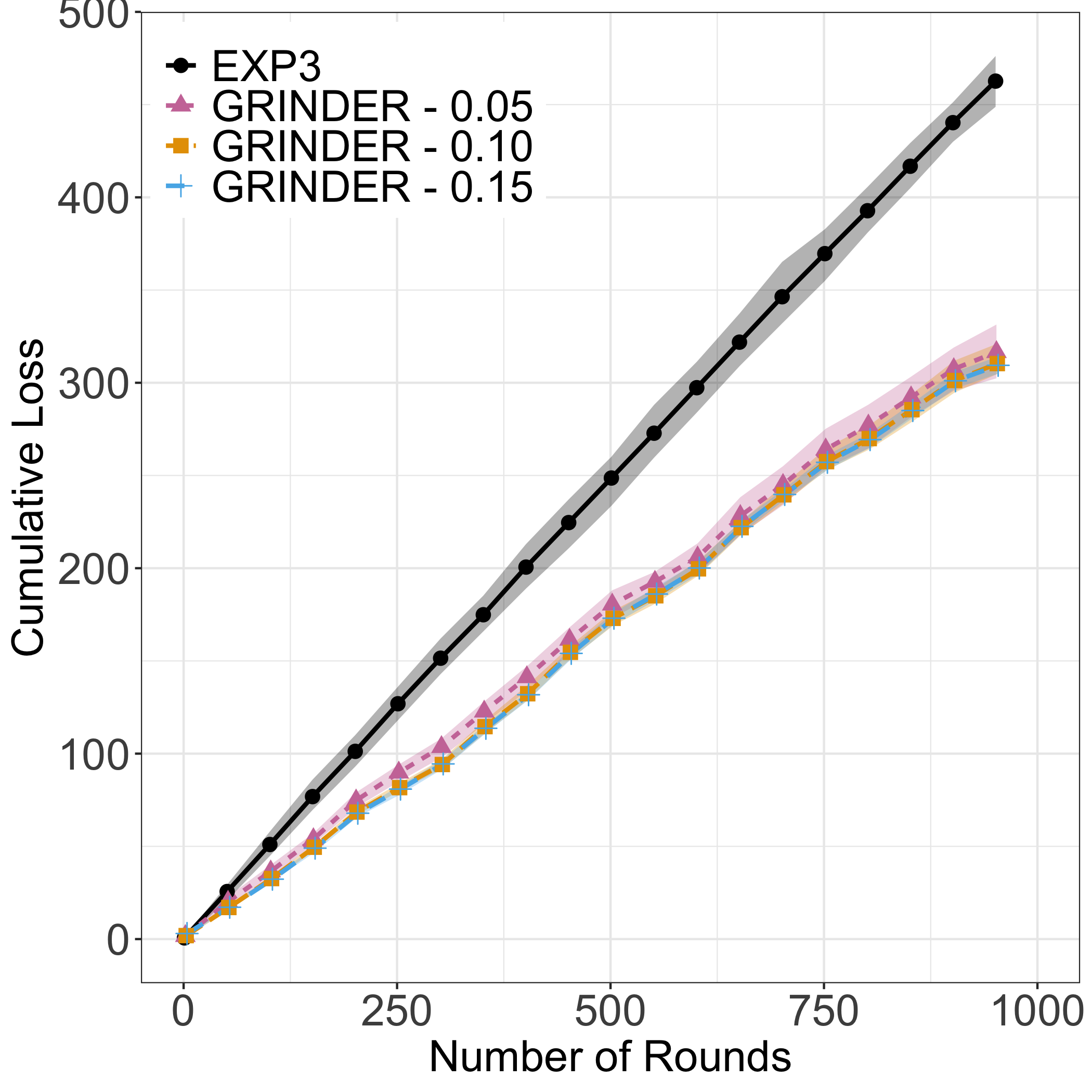}
\end{subfigure}%
\begin{subfigure}{0.33\textwidth}
    \centering
    \includegraphics[scale=0.07]{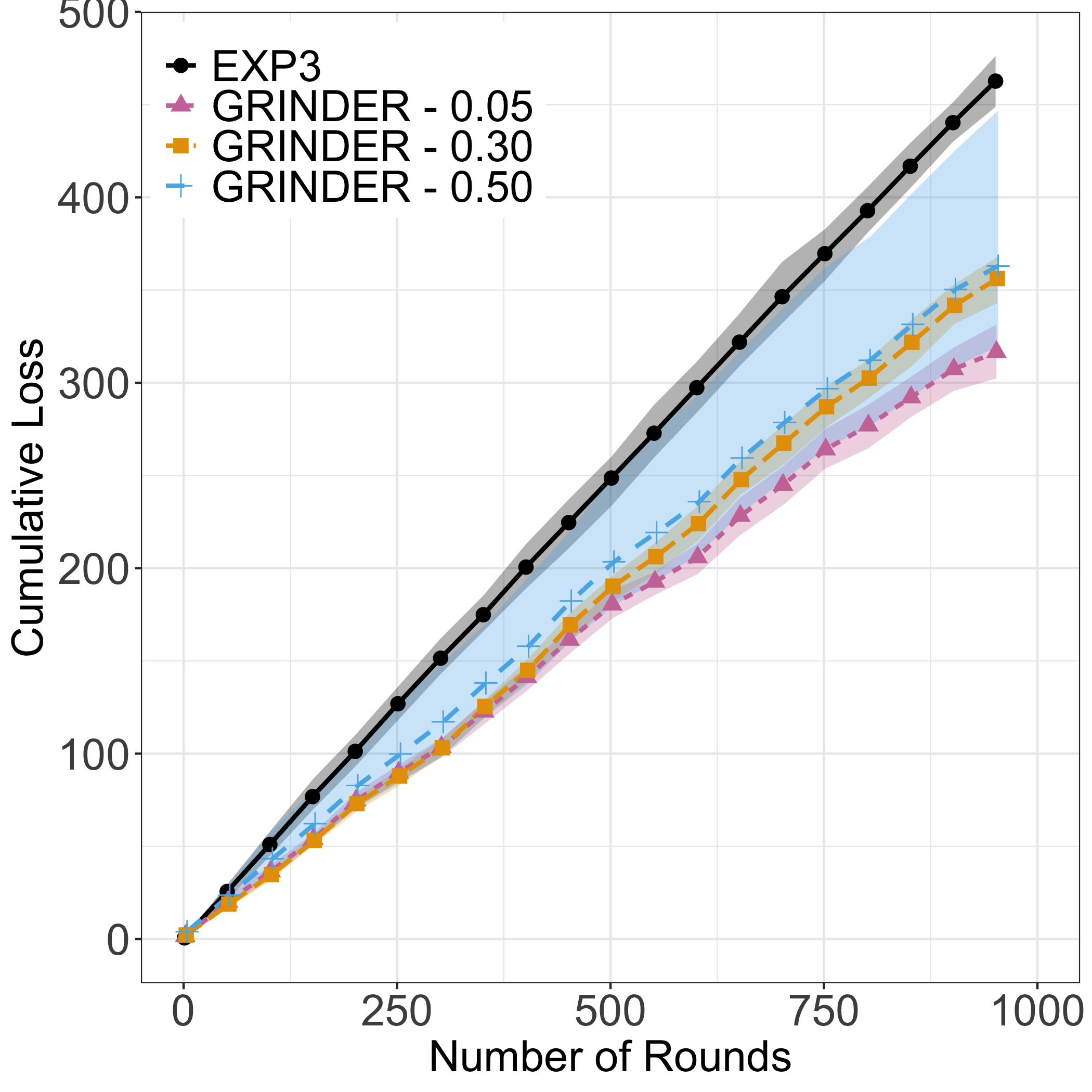}
\end{subfigure}%
\caption{\grinder vs. \texttt{EXP3} for utility function from Eq.~\eqref{eq:ut2}. From left to right: discrete $\calA$ (accurate and regression oracle), continuous $\calA$ with $\delta = 0.05, 0.10, 0.15$ and continuous $\calA$ with $\delta = 0.05, 0.3, 0.5$. Solid lines correspond to average regret/loss, and opaque bands correspond to $10$th and $90$th percentile.}\label{fig:linear-distr1}
\end{figure}

\begin{figure}[htbp]
\centering
\begin{subfigure}{0.33\textwidth}
    \includegraphics[scale=0.07]{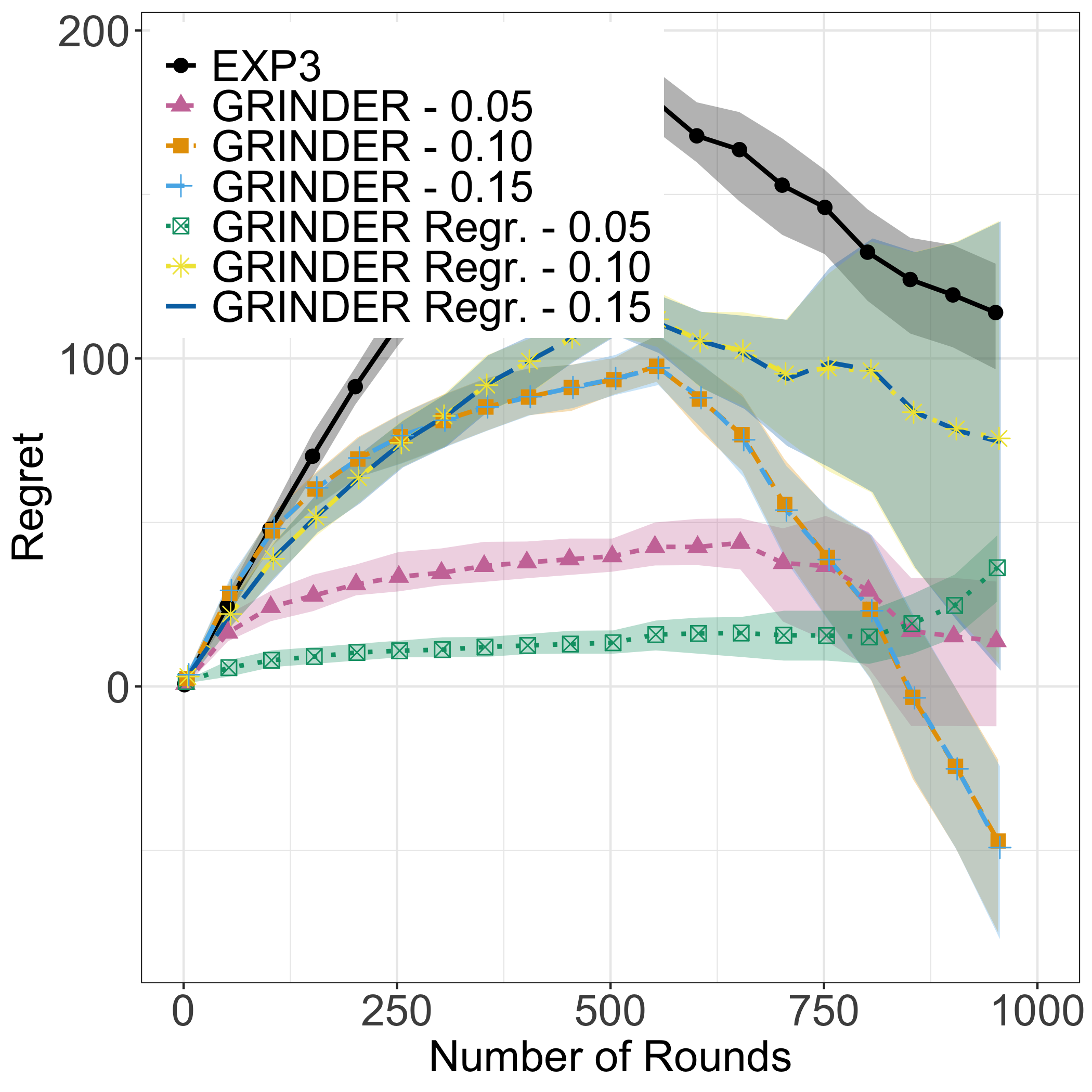}
\end{subfigure}%
\begin{subfigure}{0.33\textwidth}
    \centering
    \includegraphics[scale=0.07]{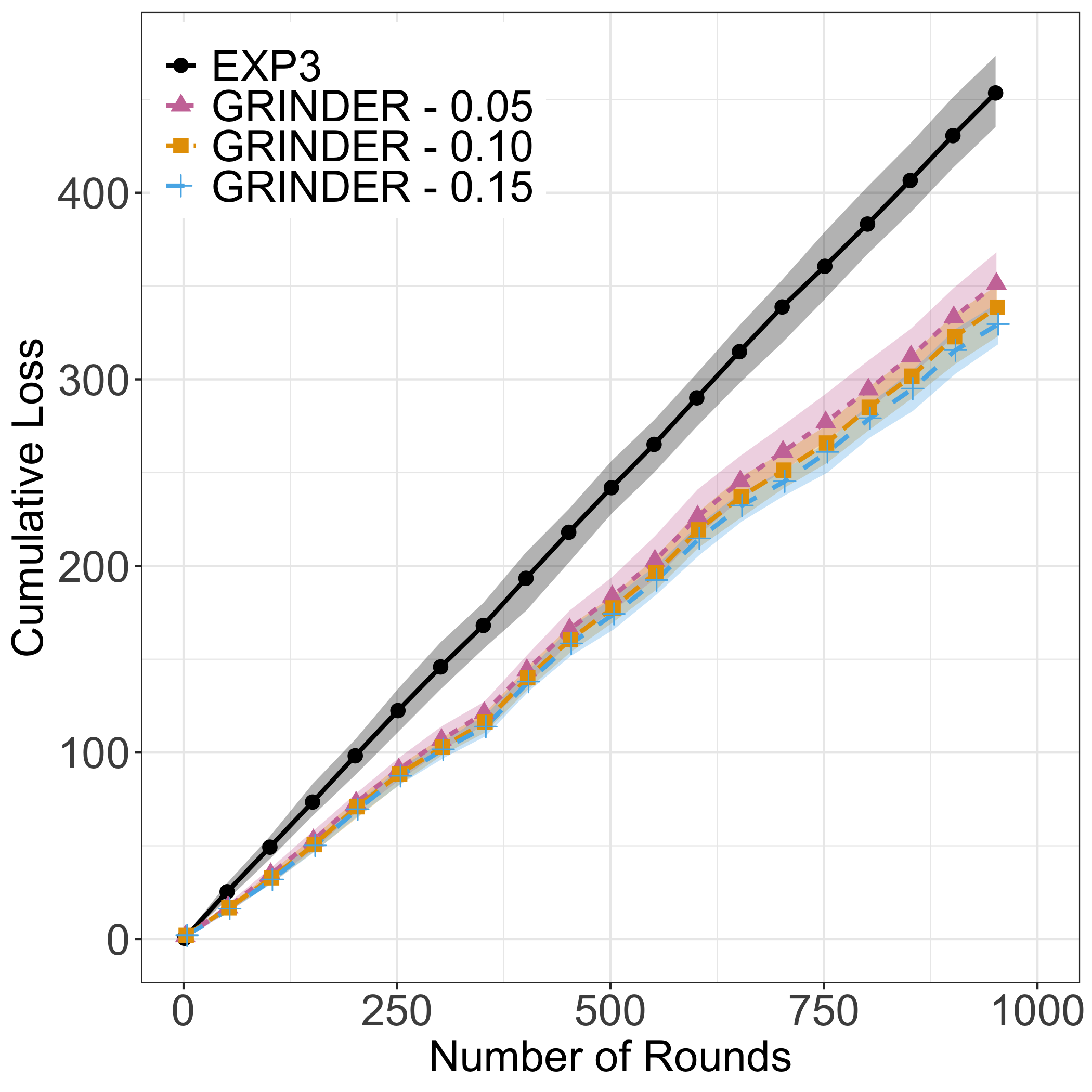}
\end{subfigure}%
\begin{subfigure}{0.33\textwidth}
    \centering
    \includegraphics[scale=0.07]{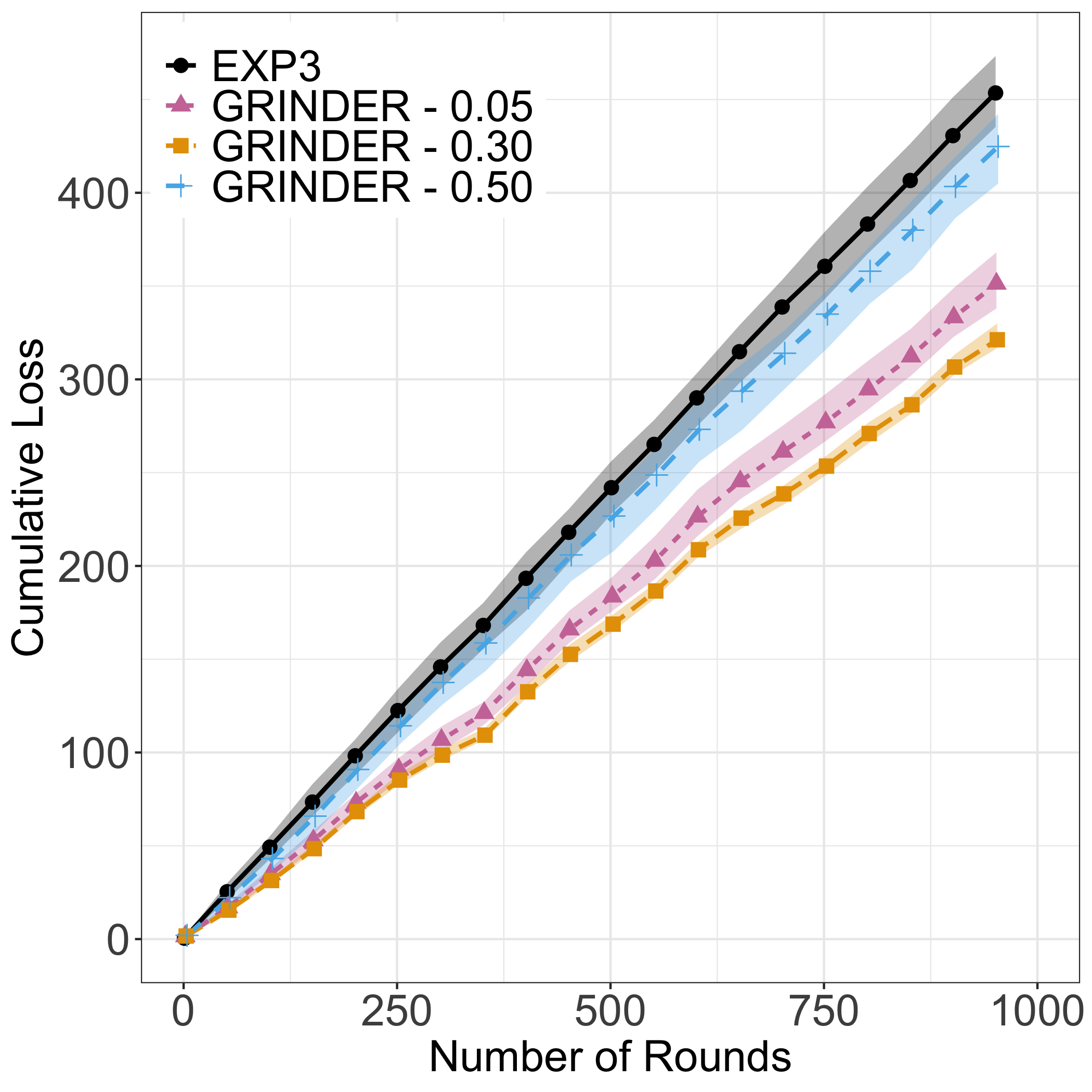}
\end{subfigure}%
\caption{\grinder vs. \texttt{EXP3} for ``harder'' distribution of labels. From left to right: discrete $\calA$ (accurate and regression oracle), continuous $\calA$ with $\delta = 0.05, 0.10, 0.15$ and continuous $\calA$ with $\delta = 0.05, 0.3, 0.5$. Solid lines correspond to average regret/loss, and opaque bands correspond to $10$th and $90$th percentile.}\label{fig:linear-distr2}
\end{figure}

Finally, in Fig.~\ref{fig:linear-distr2} we present the results of our simulations of running \grinder against \texttt{EXP3}, when the agents' utility function is defined by Equation~\eqref{eq:ut2}, and the distribution of labeled points is the following: the $+1$ labeled points are drawn as $\vx_t \sim (\calN(0.6,0.4), \calN(0.4,0.6))$ and the $-1$ labeled points are drawn from $\vx_t \sim (\calN(0.4, 0.6), \calN(0.6,0.4))$. We note that while \grinder still outperforms \texttt{EXP3} its performance has become worse than what we saw in Fig.~\eqref{fig:linear-distr1}. This is due to the fact that in this new distribution of points creates much higher overlap of labels and there are fewer points for which a perfect linear classifier exists. This is also exhibited by the fact that \texttt{EXP3}'s performance is getting better in the horizon of $T$ rounds compared to any single fixed action.

\end{document}